\documentclass[sigconf]{acmart}

\usepackage{balance}
\usepackage{booktabs} 
\usepackage{subcaption}
\usepackage[ruled]{algorithm2e}
\usepackage{tikz}
\usepackage{bbm}
\usepackage{booktabs}
\usepackage{cleveref}
\usepackage{url}
\usepackage[skip=0.5pt]{caption}
\def\BibTeX{{\rm B\kern-.05em{\sc i\kern-.025em b}\kern-.08emT\kern-.1667em\lower.7ex\hbox{E}\kern-.125emX}}

\copyrightyear{2019} 
\acmYear{2019} 
\setcopyright{acmlicensed}
\acmConference[KDD '19]{The 25th ACM SIGKDD Conference on Knowledge Discovery and Data Mining}{August 4--8, 2019}{Anchorage, AK, USA}
\acmBooktitle{The 25th ACM SIGKDD Conference on Knowledge Discovery and Data Mining (KDD '19), August 4--8, 2019, Anchorage, AK, USA}
\acmPrice{15.00}
\acmDOI{10.1145/3292500.3330891}
\acmISBN{978-1-4503-6201-6/19/08}

\newcommand{\calO}{\mathcal{O}}
\newcommand{\calL}{\mathcal{L}}
\newcommand{\bbR}{\mathbb{R}}
\newcommand{\bbE}{\mathbb{E}}
\newcommand{\diag}{\text{diag}}
\newcommand{\tr}{\text{trace}}

\newcommand{\ErdosRenyi}{Erd\H{o}s-R\'enyi}

\definecolor{burgundy}{RGB}{140,0,26}

\begin{document}
\title{Network Density of States}

\author{Kun Dong}
\affiliation{%
  \institution{Cornell University}
  \streetaddress{}
  \city{Ithaca}
  \state{New York}
  \postcode{14853}
}
\email{kd383@cornell.edu}

\author{Austin R.~Benson}
\affiliation{%
  \institution{Cornell University}
  \streetaddress{}
  \city{Ithaca}
  \state{New York}
  \postcode{14853}
}
\email{arb@cs.cornell.edu}

\author{David Bindel}
\affiliation{%
  \institution{Cornell University}
  \streetaddress{}
  \city{Ithaca}
  \state{New York}
  \postcode{14853}
}
\email{bindel@cornell.edu}

\renewcommand{\shortauthors}{K. Dong et al.}

\begin{abstract}
Spectral analysis connects graph structure to the eigenvalues and eigenvectors
of associated matrices.  Much of spectral graph theory descends directly from
spectral geometry, the study of differentiable manifolds through the spectra of
associated differential operators.  But the translation from spectral geometry
to spectral graph theory has largely focused on results involving only a few
extreme eigenvalues and their associated eigenvalues.  Unlike in geometry, the
study of graphs through the overall distribution of eigenvalues --- the {\em
spectral density} --- is largely limited to simple random graph models.  The
interior of the spectrum of real-world graphs remains largely unexplored,
difficult to compute and to interpret.

In this paper, we delve into the heart of spectral densities of real-world
graphs.  We borrow tools developed in condensed matter physics, and add novel
adaptations to handle the spectral signatures of common graph motifs.  The
resulting methods are highly efficient, as we illustrate by computing spectral
densities for graphs with over a billion edges on a single compute node. Beyond
providing visually compelling fingerprints of graphs, we show how the 
estimation of spectral densities facilitates the computation of many common
centrality measures, and use spectral densities to estimate meaningful
information about graph structure that cannot be inferred from the extremal
eigenpairs alone.
\end{abstract}

%
%



\maketitle

\section{Introduction}
\label{sec:introduction}

Spectral theory is a powerful analysis tool in graph theory~
\cite{ctekovic1998spectra,chung1997spectral,chung2006complex}, geometry~
\cite{chavel1984eigenvalues}, and physics~\cite{jackson2006math}. One follows
the same steps in each setting:
\begin{itemize}
\item Identify an object of interest, such as a graph or manifold;
\item Associate the object with a matrix or operator, often the generator
  of a linear dynamical system or the Hessian of a quadratic form over
  functions on the object; and
\item Connect spectral properties of the matrix or operator to
  structural properties of the original object.
\end{itemize}
In each case, the {\em complete} spectral decomposition is enough to
recover the original object; the interesting results relate structure
to {\em partial} spectral information.

Many spectral methods use extreme eigenvalues and associated eigenvectors. 
These are easy to compute by standard methods, and are easy to interpret in
terms of the asymptotic behavior of dynamical systems or the solutions to
quadratic optimization problems with quadratic constraints.  Several network
centrality measures, such as PageRank~\cite{page1999pagerank}, are expressed via
the stationary vectors of transition matrices, and the rate of convergence to
stationarity is bounded via the second-largest eigenvalue.  In geometry and
graph theory, Cheeger's inequality relates the second-smallest eigenvalue of a
Laplacian or Laplace-Beltrami operator to the size of the smallest bisecting
cut~\cite{cheeger1969lower,mohar1989isoperimetric}; in the graph setting, the
associated eigenvector (the Fiedler vector) is the basis for spectral algorithms
for graph partitioning~\cite{pothen1990partitioning}.  Spectral algorithms for
graph coordinates and clustering use the first few eigenvectors of a transition
matrix or (normalized) adjacency or Laplacian~\cite{belkin2001laplacian,
ng2002spectral}. For a survey of such approaches in network science, we refer to~\cite{chung2006complex}.

Mark Kac popularized an alternate approach to spectral analysis in an expository
article~\cite{kac1966hear} in which he asked whether one can determine the shape
of a physical object (Kac used a drum as an example) given the spectrum of the
Laplace operator; that is, can one ``hear'' the shape of a drum?  One can ask a
similar question in graph theory: can one uniquely determine the structure of a
network from the spectrum of the Laplacian or another related matrix?  Though
the answer is negative in both cases~\cite{gordon1992cannot,
ctekovic1998spectra}, the spectrum is enormously informative even without
eigenvector information.  Unlike the extreme eigenvalues and vectors,
eigenvalues deep in the spectrum are difficult to compute and to interpret, but
the overall distribution of eigenvalues --- known as the spectral density or
density of states --- provides valuable structural information.  For example,
knowing the spectrum of a graph adjacency matrix is equivalent to knowing
$\tr(A^k)$, the number of closed walks of any given length $k$. In some cases,
one wants {\em local} spectral densities in which the eigenvalues also have
positive weights associated with a location. Following Kac, this would give us
not only the frequencies of a drum, but also amplitudes based on where the drum
is struck.  In a graph setting, the local spectral density of an adjacency
matrix at node $j$ is equivalent to knowing $(A^k)_{jj}$, the number of closed
walks of any given length $k$ that begin and end at the node.

Unfortunately, the analysis of spectral densities of networks has been
limited by a lack of scalable algorithms.  While the normalized
Laplacian spectra of \ErdosRenyi\ random graphs have an approximately
semicircular distribution~\cite{wigner1958distribution}, and the spectral
distributions for other popular scale-free and small-world random graph models
are also known~\cite{farkas2001spectra}, there has been relatively little work
on computing spectral densities of large ``real-world'' networks. Obtaining the
full eigendecomposition is $\calO(N^3)$ for a graph with $N$ nodes, which is
prohibitive for graphs of more than a few thousand nodes.  In prior work,
researchers have employed methods, such as thick-restart Lanczos, that still do
not scale to very large graphs~\cite{farkas2001spectra}, or heuristic
approximations with no convergence analysis~\cite{banerjee2008spectrum}. It is
only recently that clever computational methods were developed simply to 
\emph{test} for hypothesized power laws in the spectra of large real-world
matrices by computing only \emph{part} of the spectrum~
\cite{eikmeier2017revisiting}.

In this paper, we show how methods used to
study densities of states in condensed matter physics~\cite{weisse2006kernel} can be
used to study spectral densities in networks.  We study these methods for both
the {\em global} density of states and for {\em local} densities of states
weighted by specific eigenvector components. We adapt these methods to take
advantage of graph-specific structure not present in most physical systems, and
analyze the stability of the spectral density to perturbations as well as the
convergence of our computational methods.  Our methods are remarkably efficient,
as we illustrate by computing densities for graphs with billions of edges and
tens of millions of nodes on a single cloud compute node. We use our methods for
computing these densities to create compelling visual fingerprints that summarize a graph.
We also show how the density of states reveals graph properties that are not
evident from the extremal eigenvalues and eigenvectors alone, and use it as a
tool for fast computation of standard measures of graph connectivity and node
centrality. This opens the door for the use of complete spectral information
as a tool in large-scale network analysis.

\section{Background}
\label{sec:background}

\subsection{Graph Operators and Eigenvalues}

We consider weighted, undirected graphs $G = (V, E)$ with vertices 
$V=\{v_1,\cdots, v_N\}$ and edges $E \subseteq V \times V$. The weighted adjacency
matrix $A\in\bbR^{N\times N}$ has entries $a_{ij} > 0$ to give the weight of an
edge $(i,j) \in E$ and $a_{ij} = 0$ otherwise. The degree matrix $D\in\bbR^
{N\times N}$ is the diagonal matrix of weighted node degrees, i.e. $D_{ii} =
\sum_j a_{ij}$. Several of the matrices in spectral graph theory are defined in
terms of $D$ and $A$.  We describe a few of these below, along with their
connections to other research areas. For each operator, we let $\lambda_1 \leq
\ldots \leq \lambda_N$ denotes the eigenvalues in ascending order.

\noindent{\bf Adjacency Matrix: $\pmb{A}$.} 
Many studies on the spectrum of $A$ originate from random matrix theory where
$A$ represents a random graph model. In these cases, the limiting behavior of
eigenvalues as $N\to\infty$ is of particular interest. Besides the growth of
extremal eigenvalues~\cite{chung1997spectral}, Wigner's semicircular law is 
the most renowned result about the spectral distribution of the adjacency
matrix~\cite{wigner1958distribution}. When the edges are i.i.d.~random variables
with bounded moments, the density of eigenvalues within a range converges to a
semicircular distribution. One famous graph model of this type is the
\ErdosRenyi\ graph, where $a_{ij} = a_{ji} = 1$ with probability $p < 1$, and
$0$ with probability $1-p$. Farkas et al.~\cite{farkas2001spectra} has extended
the semicircular law by investigating the spectrum of scale-free and small-world
random graph models. They show the spectra of these random graph models relate
to geometric characteristics such as the number of cycles and the degree
distribution.

\noindent{\bf Laplacian Matrix: $\pmb{L = D - A}$.} 
The Laplace operator arises naturally from the study of dynamics in both
spectral geometry and spectral graph theory. The continuous Laplace operator
and its generalizations are central to the description of physical systems
including heat diffusion~\cite{mckean1972selberg}, wave propagation~
\cite{levy2006laplace}, and quantum mechanics~\cite{ducastelle1970moments}. It
has infinitely many non-negative eigenvalues, and Weyl's law~
\cite{weyl1911asymptotische} relates their asymptotic distribution to the volume
and dimension of the manifold. On the other hand, the discrete Laplace matrix
appears in the formulation of graph partitioning problems. If $f \in \{ \pm 1
\}^N$ is an indicator vector for a partition $V = V_+ \cup V_-$, then $f^T L
f/4$ is the number of edges between $V_+$ and $V_{-}$, also known as the cut
size. $L$ is a positive-semidefinite matrix with the vector of all ones as a
null vector. The eigenvalue $\lambda_2$, called the {\em algebraic
connectivity}, bounds from below the smallest bisecting cut size; $\lambda_2 =
0$ if and only if the graph is disconnected. In addition, eigenvalues of $L$
also appear in bounds for vertex connectivity ($\lambda_2$)~
\cite{cvetkovic2009introduction}, minimal bisection ($\lambda_2$)~
\cite{donath2003lower}, and maximum cut ($\lambda_N$)~\cite{trevisan2012max}.

\noindent{\bf Normalized Laplacian Matrix: $\pmb{\overline{L} = I - D^{-1/2}AD^
{-1/2}}$.} We will also mention the normalized adjacency matrix $\overline{A} =
D^{-1/2}AD^{-1/2}$ and graph random walk matrix $P = D^{-1}A$ here, because
these matrices have the same eigenvalues as $\bar{L}$ up to a shift. The
connection to some of the most influential results in spectral geometry is
established in terms of eigenvalues and eigenvectors of normalized Laplacian. A
prominent example is the extension of Cheeger's inequality to the discrete
case, which relates the set of smallest conductance $h(G)$ (the Cheeger
constant)  to the second smallest eigenvalue of the normalized Laplacian,
$\lambda_2(\overline{L})$~\cite{montenegro2006mathematical}:
\[
\lambda_2(\overline{L})/2\leq
h(G) = \min_{S\subset V}\frac{\lvert \{(i,j)\in E, i\in S, j\notin S\} \rvert
  }{\min(\text{vol}(S), \text{vol}(V\backslash S))}
\leq \sqrt{2\lambda_2(\overline{L})} ,
\]
where $\text{vol}(T) = \sum_{i \in T}\sum_{j=1}^{N}a_{ij}$. Cheeger's inequality
offers crucial insights and powerful techniques for understanding popular
spectral graph algorithms for partitioning~\cite{mcsherry2001spectral} and
clustering~\cite{ng2002spectral}. It also plays a key role in analyzing the
mixing time of Markov chains and random walks on a graph
~\cite{Mihail-1989-Markov,Sinclair-1989-Markov}. For all these problems,
extremal eigenvalues again emerge from relevant optimization formulations.

\subsection{Spectral Density (Density of States --- DOS)}
Let $H = \bbR^{N\times N}$ be any symmetric graph matrix with an
eigendecomposition $H = Q\Lambda Q^T$, where $\Lambda = \diag 
(\lambda_1,\cdots,\lambda_N)$ and $Q = [q_1,\cdots, q_N]$ is orthogonal. The
spectral density induced by $H$ is the generalized function
\begin{equation}\label{eqn:dos}
\mu(\lambda) = \frac{1}{N}\sum_{i=1}^N \delta(\lambda - \lambda_i),
\quad \int f(\lambda) \mu(\lambda) = \tr(f(H))
\end{equation}
where $\delta$ is the Dirac delta function and $f$ is any analytic test
function. The spectral density $\mu$ is also referred to as the \emph{density
of states} (DOS) in the condensed matter physics literature~
\cite{weisse2006kernel}, as it describes the number of states at different
energy levels. For any vector $u\in\bbR^N$, the local density of states (LDOS)
is
\begin{equation}\label{eqn:ldos}
\mu(\lambda; u) = \sum_{i=1}^N|u^Tq_i|^2\delta(\lambda-\lambda_i),
\quad
\int f(\lambda)\mu(\lambda; u) = u^T f(H) u.
\end{equation}
Most of the time, we are interested in the case $u=e_k$ where $e_k$ is the $k$th
standard basis vector---this provides the spectral information about a
particular node. We will write $\mu_k(\lambda) = \mu(\lambda; e_k)$ for the
pointwise density of states (PDOS) for node $v_k$. It is noteworthy $|e_k^Tq_i|
= |q_i(k)|$ gives the magnitude of the weight for $v_k$ in the $i$-th
eigenvector, thereby the set of $\{\mu_k\}$ encodes the entire spectral
information of the graph up to sign differences. These concepts can be easily
extended to directed graphs with asymmetric matrices, for which the eigenvalues
are replaced by singular values, and eigenvectors by left/right singular
vectors.

Naively, to obtain the DOS and LDOS requires computing all eigenvalues and
eigenvectors for an $N$-by-$N$ matrix, which is infeasible for large graphs.
Therefore, we turn to algorithms that approximate these densities. Since the DOS
is a generalized function, it is important we specify how the estimation is
evaluated. One choice is to treat $\mu$ (or $\mu_k$) as a distribution, and
measure its  approximation error with respect to a chosen function space
$\calL$. For  example, when $\calL$ is the set of Lipschitz continuous functions
taking the value 0 at 0, the error for estimated $\widetilde{\mu}$ is in the
Wasserstein distance (a.k.a.\ earth-mover distance)~\cite{kantorovich1958space}
\begin{equation}\label{eqn:waisserstein}
W_1(\mu,\widetilde{\mu}) = \sup\Big\{\int (\mu(\lambda)-\widetilde{\mu}(\lambda
))f(\lambda)d\lambda : \text{Lip}(f)\leq 1\Big\}.
\end{equation}
This notion is particularly useful when $\mu$ is integrated against in 
applications such as computing centrality measures.

On the other hand, we can regularize $\mu$ with a mollifier $K_\sigma$ (i.e.,
a smooth approximation of the identity function):
\begin{equation}\label{eqn:reg_dos}
(K_\sigma\ast\mu)(\lambda) = \int_\bbR \sigma^{-1} K\left(\frac{\lambda-\nu}{
\sigma}\right)\mu(\nu)d\nu
\end{equation}
A simplified approach is numerically integrating $\mu$ over small intervals of
equal size to generate a spectral histogram. The advantage is the error is now
easily measured and visualized in the $L_\infty$ norm. For example, Figure 
\ref{fig:caida} shows the exact and approximated spectral histogram for the
normalized adjacency matrix of an Internet topology. 
\begin{figure}
  \begin{subfigure}{0.235\textwidth}
    \centering  
    \includegraphics[width=\textwidth,trim = .4cm 0.5cm 3.5cm 1.3cm,clip]{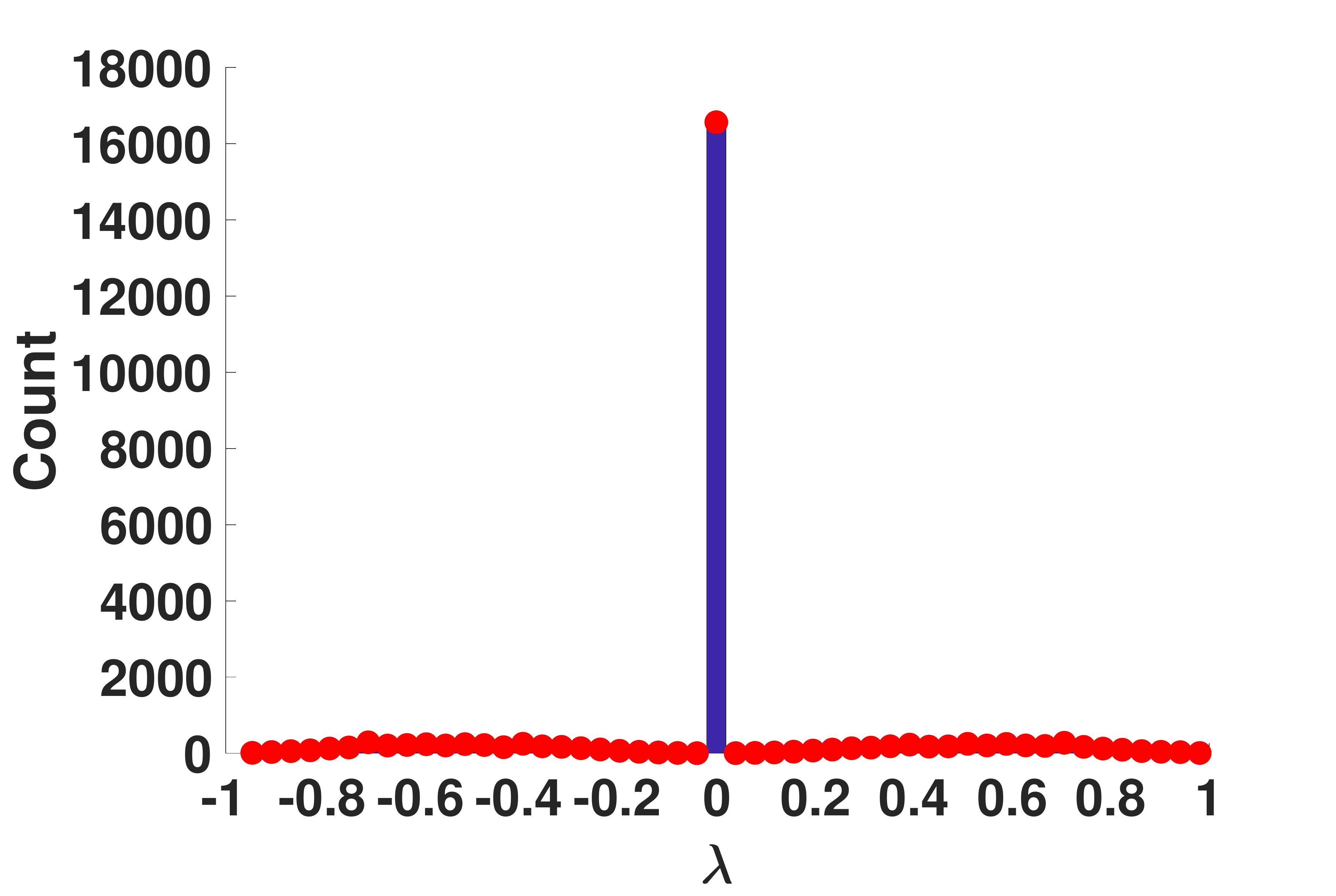}
    \caption{Spectral Histogram}
    \label{fig:caida_full}
  \end{subfigure}
  \begin{subfigure}{0.235\textwidth}
    \centering
    \includegraphics[width=\textwidth,trim = .4cm 0.5cm 3.5cm 1.3cm,clip]{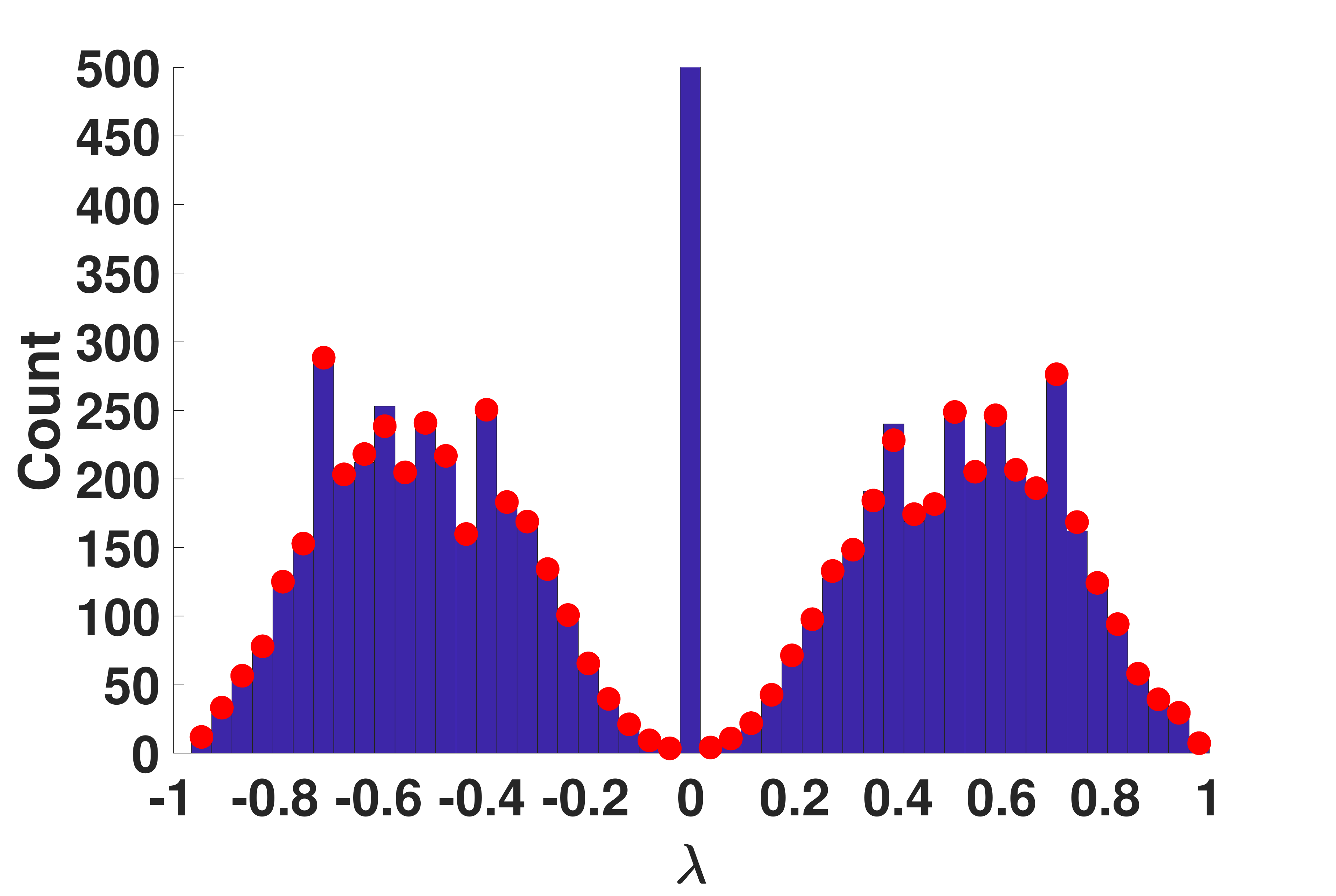}
    \caption{Zoom-in View}
    \label{fig:caida_zoom}
  \end{subfigure}
  \caption{Spectral histogram for the normalized adjacency matrix
  for the CAIDA autonomous systems graph~\cite{caida2012}, 
  an Internet topology with $22965$ nodes and $47193$ edges. Blue
  bars are the real spectrum, and red points are the approximated heights. (\ref{fig:caida_full}) contains high multiplicity around eigenvalue $0$, so (\ref{fig:caida_zoom}) zooms in to height between $[0,500]$.}
  \label{fig:caida}
\end{figure}

\section{Methods}
\label{sec:methods}

The density of states plays a significant role in understanding electronic band
structure in solid state physics, and so several methods have been proposed in
that literature to estimate spectral densities. We review two such methods: the
kernel polynomial method (KPM) which involves a polynomial expansion of the
DOS/LDOS, and the Gauss Quadrature via Lanczos iteration (GQL). These methods
have not previously been applied in the network setting, though Cohen-Steiner et
al.~\cite{cohen2018approximating} have independently invented an approach
similar to KPM for the global DOS alone, albeit using a less numerically stable
polynomial basis (the power basis associated with random walks). We then
introduce a new direct nested dissection method for LDOS, as well as new
graph-specific modifications to improve the convergence of the KPM and GQL
approaches.

Throughout this section, $H$ denotes any symmetric matrix.

\subsection{Kernel Polynomial Method (KPM)}\label{subsec:kpm}

The Kernel Polynomial Method (KPM)~\cite{weisse2006kernel} approximates
the spectral density through an expansion in the dual basis of an orthogonal
polynomial basis. Traditionally, the Chebyshev basis $\{T_m\}$ is used because
of its connection to the best polynomial interpolation. Chebyshev approximation
requires the spectrum to be supported on the interval $[-1,1]$ for numerical
stability. However, this condition can be satisfied by any graph matrix after
shifting and rescaling:
\begin{equation}\label{eqn:shiftscale}
\widetilde{H} = \frac{2H - (\lambda_{\max}(H)+\lambda_{\min}(H))}{\lambda_{\max}(H) -
\lambda_{\min}(H)}
\end{equation}
We can compute these extremal eigenvalues efficiently for our sparse matrix $H$,
so the pre-computation is not an issue~\cite{Parlett-1984-sparse}.

The Chebyshev polynomials $T_m(x)$ satisfy the recurrence
\begin{equation}\label{eq:cheb_3term}
T_0(x) = 1,\; T_1(x) = x,\; T_{m+1}(x) = 2xT_m(x) - T_{m-1}(x).
\end{equation}
They are orthogonal with respect to $w(x) =  2 / [(1+\delta_{0n})\pi
\sqrt{1-x^2}]$:
\begin{equation} \label{eqn:orth}
\int_{-1}^1 w(x)T_m(x)T_n(x)dx = \delta_{mn}.
\end{equation}
(Here and elsewhere, $\delta_{ij}$ is the Kronecker delta: $1$ if $i = j$ and
$0$ otherwise.) Therefore, $T_m^\ast(x) = w(x)T_m(x)$ also forms the dual
Chebyshev basis. Using (\ref{eqn:orth}), we can expand our DOS $\mu(\lambda)$ as
\begin{gather} \label{eqn:doscoeff}
\mu(x) = \sum_{m=1}^\infty d_mT^\ast_m(\lambda) \\
d_m = \int_{-1}^1T_m(\lambda)\mu(\lambda)d\lambda = \frac{1}{N}\sum_{i=1}^N 
T_m(\lambda_i) = \frac{1}{N}\tr(T_m(H)),
\end{gather}
Here, $T_m(H)$ is the $m$th Chebyshev polynomial of the matrix $H$. The last
equality comes from the spectral mapping theorem, which says that taking a
polynomial of $H$ maps the eigenvalues by the same polynomial. Similarly, we
express the PDOS $\mu_k(\lambda)$ as \begin{equation}\label{eqn:pdoscoeff}
d_{mk} = \int_{-1}^1 T_m(\lambda)\mu_k(\lambda)d\lambda = \sum_{i=1}^N |q_i(k)|
^2T_m(\lambda_i) = T_m(H)_{kk}.
\end{equation}

We want to efficiently extract the diagonal elements of the matrices $\{T_m
(H)\}$ without forming them explicitly; the key idea is to apply the stochastic
trace/diagonal estimation, proposed by Hutchinson~
\cite{hutchinson1990stochastic} and Bekas et al.~\cite{bekas2007estimator}.
Given a random probe vector $z$ such that $z_i$'s are i.i.d.~with mean $0$ and
variance $1$,
\begin{equation}\label{eq:probe_trace}
\bbE[z^THz] = \sum_{i,j}H_{ij}\bbE[z_iz_j] = \tr(H)
\end{equation}
\begin{equation}\label{eq:probe_diag}
\bbE[z\odot Hz] = \diag(H)
\end{equation}
where $\odot$ represents the Hadamard (elementwise) product. Choosing $N_z$
independent probe vectors $Z_j$, we obtain the unbiased estimator
\begin{equation*}
\tr(H) = \bbE[z^THz] \approx \frac{1}{N_z}\sum_{j=1}^{N_z}Z_j^THZ_j 
\end{equation*}
and similarly for the diagonal. Avron and Toledo~\cite{avron2011randomized}
review many possible choices of probes for \cref{eq:probe_trace,eq:probe_diag};
a common choice is vectors with independent standard normal entries. Using the
Chebyshev recurrence (\cref{eq:cheb_3term}), we can compute the sequence $T_j
(H)z$ for each probe at a cost of one matrix-vector product per term, for a
total cost of $O(|E| N_z)$ time per moment $T_m(H)$.

In practice, we only use a finite number of moments rather than an 
infinite expansion. The number of moments required depends on the convergence
rate of the Chebyshev approximation for the class of functions DOS/LDOS is
integrated with. For example, the approximation error decays exponentially for
test functions that are smooth over the spectrum
\cite{trefethen2013approximation}, so only a few moments are needed. On the
other hand, such truncation leads to Gibbs oscillations that cause error in the
interpolation ~\cite{Trefethen-2013-ATAP}. However, to a large extent, we can 
use smoothing techniques such as Jackson damping to resolve this issue 
~\cite{jackson1911genauigkeit} (we will formalize this in 
\cref{thm:Jackson_damping}).

\subsection{Gauss Quadrature and Lanczos (GQL)}
Golub and Meurant developed the well-known Gauss Quadrature and Lanczos (GQL)
algorithm to approximate bilinear forms for smooth functions of a matrix~
\cite{golub1997matrices}.  Using the same stochastic estimation from \S
\ref{subsec:kpm}, we can also apply GQL to compute DOS.

For a starting vector $z$ and graph matrix $H$, Lanczos iterations after
$M$ steps produce a decomposition
\begin{equation*}
HZ_M = Z_M^T\Gamma_M + r_Me_M^T
\end{equation*}
where $Z_M^TZ_M = I_M$, $Z_M^Tr_M = 0$, and $\Gamma_M$ tridiagonal. GQL 
approximates $z^Tf(H)z$ with $\|z\|^2e_1^Tf(T_M)e_1$, implying
\begin{equation*}
z^Tf(H)z = \sum_{i=1}^N |z^Tq_i|^2f(\lambda_i) \approx \|z\|^2\sum_{i=1}^M |p_{
i1}|^2f(\tau_i)
\end{equation*}
where $(\tau_1,p_1)\cdots,(\tau_M, p_M)$ are the eigenpairs of $\Gamma_M$. 
Consequently,
\begin{equation*}
\|z\|^2\sum_{i=1}^M |p_{i1}|^2\delta(\lambda-\tau_i)
\end{equation*}
approximates the LDOS $\mu(\lambda; z)$.

Building upon the stochastic estimation idea and the invariance of probe 
vectors under orthogonal transformation, we have
\[
\bbE[\mu(\lambda;z)] = \sum_{i=1}^N \delta(\lambda-\lambda_i) = N\mu(\lambda)
\]
Hence 
\begin{equation*}
\mu(\lambda)\approx\sum_{i=1}^M |p_{i1}|^2\delta(\lambda-\tau_i).
\end{equation*}
The approximate generalized function is exact when applied to polynomials of
degree $\leq 2M+1$. Furthermore, if we let $z = e_k$ then GQL also provides an
estimation for the PDOS $\mu_k(\lambda)$. Estimation from GQL can also be
converted to Chebyshev moments if needed.

\subsection{Nested Dissection (ND)}

The estimation error via Monte Carlo method intrinsically decays at the rate
$\calO(1 / \sqrt{N_z})$, where $N_z$ is the number of random probing vectors.
Hence, we have to tolerate the higher variance when increasing the number of
probe vectors becomes too expensive. This is particularly problematic when we
try to compute the PDOS for all nodes using the stochastic diagonal estimator.
Therefore, we introduce an alternative divide-and-conquer  method, which
computes more accurate PDOS for any set of nodes at a cost comparable to the
stochastic approach in practice.

Suppose the graph can be partitioned into two subgraphs by removal of a small
vertex separator.  Permuting the vertices so that the two partitions appear
first, followed by the separator vertices. Up to vertex permutations, we can
rewrite $H$ in block form as
\begin{equation*}
H = \begin{bmatrix}
H_{11} & 0 & H_{13}\\
0 & H_{22} & H_{23}\\
H_{13}^T & H_{23}^T & H_{33}
\end{bmatrix},
\end{equation*}
where the indices indicate the groups identities. Leveraging this structure, we
can update the recurrence relation for Chebyshev polynomials to become
\begin{equation}\label{eqn:ndcheb}
T_{m+1}(H)_{11} = 2H_{11}T_m(H)_{11} - T_{m-1}(H)_{11} + 2H_{13}T_m(H)_{31}
\end{equation}

Recursing on the partitioning will lead to a nested dissection, after which we
will use direct computation on sufficiently small sub-blocks. We denote the
indexing of each partition with $I^{(t)}_p = I^{(t)}_s\bigcup I^{
(t)}_\ell\bigcup I^{(t)}_r$, which represents all nodes in the current
partition, the separators, and two sub-partitions, respectively. For the
separators, equation \ref{eqn:ndcheb} leads to
\begin{align}\label{eqn:ndchebsep}
&T_{m+1}(H)(I^{(t)}_p,I^{(t)}_s) = 2H(I^{(t)}_p, I^{(t)}_p)T_m(H)(I^{(t)}_p,
I^{(t)}_s) \nonumber\\
&-T_{m-1}(H)(I^{(t)}_p,I^{(t)}_s) + 2\sum_{t'\in S_t}H(I^{(t)}_p,I^{(t')}_s)T_m
(H)(I^{(t')}_s, I^{(t)}_s)
\end{align}
where $S_t$ is the path from partition $t$ to the root; and for the leaf
blocks, $I^{(t)}_s = I^{(t)}_p$ in equation \ref{eqn:ndchebsep}.
The result is Algorithm 1.

\begin{algorithm}
\SetKwInput{KwInput}{Input}
\SetKwInput{KwOutput}{Output}
\DontPrintSemicolon
	\KwInput{Symmetric graph matrix $H$ with eigenvalues in $[-1,1]$}
	\KwOutput{$C\in\bbR^{N\times M}$ where $c_{ij}$ is the $j$-th Chebyshev
	moment for $i$-th node.}
	\Begin{
		Obtain partitions $\{I^{(t)}_p\}$ in a tree structure through
		multilevel nested dissection.\;
		\For{$m=1$ \textbf{to} $M$}{
			Traverse partition tree in pre-order:\;
			\quad Compute the separator columns with \cref{eqn:ndchebsep}.\;
			\quad \If{$I^{(t)}_p$ is a leaf block}{Compute the diagonal entries
			with equation (\ref{eqn:ndchebsep}).}
		}
	}
\caption{Nested Dissection for PDOS Approximation}
\label{figure:ndalg}        
\end{algorithm}

The multilevel nested dissection process itself has a well-established
algorithm by Karypis and Kumar, and efficient implementation is available
in \emph{METIS}~\cite{karypis1998fast}. Note that this approach is only viable
when the graph can be partitioned with a separator of small size. Empirically,
we observe this assumption to hold for many real-world networks. The biggest
advantage of this approach is we can very efficiently obtain PDOS estimation for
a subset of nodes with much better accuracy than KPM.

\subsection{Motif Filtering}

In many graphs, there are large spikes around particular eigenvalues; for
example, see \cref{fig:caida}. This phenomenon affects the accuracy of DOS
estimation in two ways. First, the singularity-like behavior means we need many
more moments to obtain a good approximation in polynomial basis. Secondly, due
to the equi-oscillation property of Chebyshev approximation, error in
irregularities (say, at a point of high concentration in the spectral density),
spreads to other parts of the spectrum. This is a problem in our case, as the
spectral density of real-world networks are far from uniform.

High multiplicity eigenvalues are typically related to local symmetries in a
graph. The most prevalent example is two dangling nodes attached to the same
neighbor as shown in \ref{fig:motif_0}, which accounts for most eigenvalues
around $0$ for (normalized) adjacency matrix with a localized eigenvector taking
value $+1$ on one node and $-1$ on the other. In addition, we list a few more
motifs in figure \ref{fig:motifs} that appear most frequently in real-world
graphs. All of them can be associated with specific eigenvalues, and we include
the corresponding ones in normalized adjacency matrix for our example.

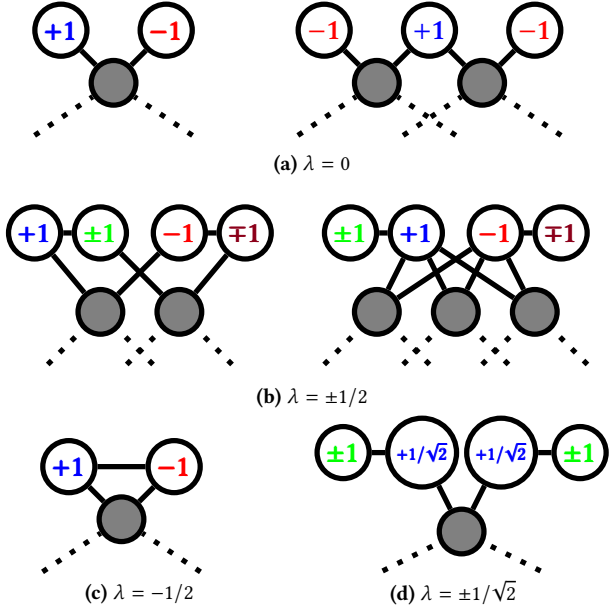
\begin{figure}
  \begin{subfigure}{0.5\textwidth}
    \scalebox{0.7}{
	\begin{tikzpicture}[remember picture, overlay]
		\node[draw,line width=1mm,circle,inner sep=1mm] (A) at (1.5,-0.5) 
		{\fontsize{16}{1}\selectfont$\pmb{\color{blue}+1}$};
		\node[draw,line width=1mm,circle,inner sep=1mm] (B) at (3.5,-0.5)
		{\fontsize{16}{1}\selectfont$\pmb{\color{red}-1}$};
		\node[draw,line width=1mm,fill=gray,circle,inner sep=3mm] (C) at 
		(2.5,-1.5) {};
		\node[draw,line width=1mm,circle,inner sep=1mm] (D) at (6.5,-0.5)
		{\fontsize{16}{1}\selectfont$\mathbf{\color{red}-1}$};
		\node[draw,line width=1mm,circle,inner sep=1mm] (E) at (8.5,-0.5)
		{\fontsize{16}{1}\selectfont$\mathbf{\color{blue}+1}$};
		\node[draw,line width=1mm,circle,inner sep=1mm] (F) at (10.5,-0.5)
		{\fontsize{16}{1}\selectfont$\mathbf{\color{red}-1}$};
		\node[draw,line width=1mm,fill=gray,circle,inner sep=3mm] (G) at 
		(7.5,-1.5) {};
		\node[draw,line width=1mm,fill=gray,circle,inner sep=3mm] (H) at 
		(9.5,-1.5) {};

		\draw[line width=1mm] (A)--(C);
		\draw[line width=1mm] (B)--(C);
		\draw[line width=1mm, loosely dashed] (1,-2.5)--(C);
		\draw[line width=1mm, loosely dashed] (4,-2.5)--(C);
		\draw[line width=1mm] (D)--(G);
		\draw[line width=1mm] (E)--(G);
		\draw[line width=1mm] (E)--(H);
		\draw[line width=1mm] (F)--(H);
		\draw[line width=1mm, loosely dashed] (6,-2.5)--(G);
		\draw[line width=1mm, loosely dashed] (9,-2.5)--(G);
		\draw[line width=1mm, loosely dashed] (8,-2.5)--(H);
		\draw[line width=1mm, loosely dashed] (11,-2.5)--(H);
	\end{tikzpicture}}
	\vspace{1.75cm}{}
	\caption{$\lambda = 0$}
	\label{fig:motif_0}
  \end{subfigure}
  \begin{subfigure}{0.5\textwidth}
    \scalebox{0.7}{
	\begin{tikzpicture}[remember picture, overlay]
		\node[draw,line width=1mm,circle,inner sep=1mm] (A) at (1,-1) 
		{\fontsize{16}{1}\selectfont$\pmb{\color{blue}+1}$};
		\node[draw,line width=1mm,circle,inner sep=1mm] (B) at (2.25,-1)
		{\fontsize{16}{1}\selectfont$\pmb{\color{green}\pm 1}$};
		\node[draw,line width=1mm,circle,inner sep=1mm] (C) at (3.75,-1)
		{\fontsize{16}{1}\selectfont$\pmb{\color{red}-1}$};
		\node[draw,line width=1mm,circle,inner sep=1mm] (D) at (5,-1)
		{\fontsize{16}{1}\selectfont$\pmb{\color{burgundy}\mp 1}$};
		\node[draw,line width=1mm,fill=gray,circle,inner sep=3mm] (E) at 
		(2.25,-2.5) {};
		\node[draw,line width=1mm,fill=gray,circle,inner sep=3mm] (F) at 
		(3.75,-2.5) {};
		\node[draw,line width=1mm,circle,inner sep=1mm] (G) at (7,-1) 
		{\fontsize{16}{1}\selectfont$\pmb{\color{green}\pm1}$};
		\node[draw,line width=1mm,circle,inner sep=1mm] (H) at (8.25,-1)
		{\fontsize{16}{1}\selectfont$\pmb{\color{blue} +1}$};
		\node[draw,line width=1mm,circle,inner sep=1mm] (I) at (9.75,-1)
		{\fontsize{16}{1}\selectfont$\pmb{\color{red}-1}$};
		\node[draw,line width=1mm,circle,inner sep=1mm] (J) at (11,-1)
		{\fontsize{16}{1}\selectfont$\pmb{\color{burgundy}\mp 1}$};
		\node[draw,line width=1mm,fill=gray,circle,inner sep=3mm] (K) at 
		(7.5,-2.5) {};
		\node[draw,line width=1mm,fill=gray,circle,inner sep=3mm] (L) at 
		(9,-2.5) {};
		\node[draw,line width=1mm,fill=gray,circle,inner sep=3mm] (M) at 
		(10.5,-2.5) {};

		\draw[line width=1mm] (A)--(E);
		\draw[line width=1mm] (C)--(E);
		\draw[line width=1mm] (A)--(B);
		\draw[line width=1mm] (C)--(D);
		\draw[line width=1mm] (B)--(F);
		\draw[line width=1mm] (D)--(F);
		\draw[line width=1mm, loosely dashed] (1.25,-3.5)--(E);
		\draw[line width=1mm, loosely dashed] (2.75,-3.5)--(F);
		\draw[line width=1mm, loosely dashed] (3.25,-3.5)--(E);
		\draw[line width=1mm, loosely dashed] (4.75,-3.5)--(F);

		\draw[line width=1mm] (G)--(H);
		\draw[line width=1mm] (I)--(J);
		\draw[line width=1mm] (H)--(K);
		\draw[line width=1mm] (H)--(L);
		\draw[line width=1mm] (H)--(M);
		\draw[line width=1mm] (I)--(K);
		\draw[line width=1mm] (I)--(L);
		\draw[line width=1mm] (I)--(M);
		\draw[line width=1mm, loosely dashed] (6.5,-3.5)--(K);
		\draw[line width=1mm, loosely dashed] (8.5,-3.5)--(K);
		\draw[line width=1mm, loosely dashed] (8,-3.5)--(L);
		\draw[line width=1mm, loosely dashed] (10,-3.5)--(L);
		\draw[line width=1mm, loosely dashed] (9.5,-3.5)--(M);
		\draw[line width=1mm, loosely dashed] (11.5,-3.5)--(M);
	\end{tikzpicture}}
	\vspace{2.5cm}
	\caption{$\lambda = \pm 1/2$}
	\label{fig:motif_pmhalf}
  \end{subfigure}
  \begin{subfigure}{0.23\textwidth}
  \scalebox{0.7}{
	\begin{tikzpicture}[remember picture, overlay]
		\node[draw,line width=1mm,circle,inner sep=1mm] (A) at (1.5,-1) 
		{\fontsize{16}{1}\selectfont$\pmb{\color{blue}+1}$};
		\node[draw,line width=1mm,circle,inner sep=1mm] (B) at (3.5,-1)
		{\fontsize{16}{1}\selectfont$\pmb{\color{red}-1}$};
		\node[draw,line width=1mm,fill=gray,circle,inner sep=3mm] (C) at 
		(2.5,-2) {};

		\draw[line width=1mm] (A)--(C);
		\draw[line width=1mm] (B)--(C);
		\draw[line width=1mm] (A)--(B);
		\draw[line width=1mm, loosely dashed] (1,-3)--(C);
		\draw[line width=1mm, loosely dashed] (4,-3)--(C);
	\end{tikzpicture}}
	\vspace{2cm}
	\caption{$\lambda = -1/2$}
	\label{fig:motif_minushalf}
  \end{subfigure}
  \begin{subfigure}{0.23\textwidth}
  \scalebox{0.7}{
	\begin{tikzpicture}[remember picture, overlay]
		\node[draw,line width=1mm,circle,inner sep=1mm] (A) at (0.75,-.75) 
		{\fontsize{16}{1}\selectfont$\pmb{\color{green}\pm1}$};
		\node[draw,line width=1mm,circle,inner sep=1mm] (B) at (2.25,-.75)
		{\fontsize{10}{1}\selectfont$\pmb{\color{blue}+1/\sqrt{2}}$};
		\node[draw,line width=1mm,circle,inner sep=1mm] (C) at (3.75,-.75)
		{\fontsize{10}{1}\selectfont$\pmb{\color{blue}+1/\sqrt{2}}$};
		\node[draw,line width=1mm,circle,inner sep=1mm] (D) at (5.25,-.75)
		{\fontsize{16}{1}\selectfont$\pmb{\color{green}\pm 1}$};
		\node[draw,line width=1mm,fill=gray,circle,inner sep=3mm] (E) at 
		(3,-2.25) {};

		\draw[line width=1mm] (A)--(B);
		\draw[line width=1mm] (B)--(E);
		\draw[line width=1mm] (C)--(D);
		\draw[line width=1mm] (C)--(E);
		\draw[line width=1mm, loosely dashed] (1.5,-3)--(E);
		\draw[line width=1mm, loosely dashed] (4.5,-3)--(E);
	\end{tikzpicture}}
	\vspace{2cm}{}
	\caption{$\lambda = \pm1/\sqrt{2}$}
	\label{fig:motif_sqrt2}
  \end{subfigure}
  \caption{Common motifs (induced subgraphs) in graph data that result in
  localized spikes in the spectral density. Each motif generates a specific
  eigenvalue with locally-supported eigenvectors. Here we uses the normalized
  adjacency matrix to represent the graph, although we can perform the same
  analysis for the adjacency, Laplacian, or normalized Laplacian (only the
  eigenvalues would be different). The eigenvectors are supported only on the
  labeled nodes.}
  \label{fig:motifs}
\end{figure}

To detect these motifs in large graphs, we deploy a randomized hashing
technique. Given a random vector $z$, the hashing weight $w=Hz$ encodes all
the neighborhood information of each node. To find node copies (left in Figure 
\ref{fig:motif_0}), we seek pairs $(i, j)$ such that $w_i=w_j$; with high
probability, this only happens when $v_i$ and $v_j$ share the same neighbors.
Similarly, all motifs in Figure \ref{fig:motifs} can be characterized by union
and intersection of neighborhood lists.

After identifying motifs, we need only approximate the (relatively smooth)
density of the remaining spectrum.  The eigenvectors associated with these
remaining non-motif eigenvalues must be constant across cycles in the canonical
decomposition of the associated permutations.  Let $P \in \bbR^{N \times r}$
denote an orthonormal basis for the space of such vectors formed from columns of
the identity and (normalized) indicators for nodes cyclically permuted by the
motif.  The matrix $H_r = P^T H P$ then has identical eigenvalues to $H$, except
with all the motif eigenvalues omitted. We may form $H_r$ explicitly, as it has
the same sparsity structure as $H$ but with a supernode replacing the nodes in
each instance of a motif cycle; or we can achieve the same result by replacing
each random probe $Z$ with the projected probe $Z_r = PP^T Z$ at an additional
cost of $O(N_{\mathrm{motif}})$ per probe, where $N_{\mathrm{motif}}$ is the
number of nodes involved in motifs.

The motif filtering method essentially allow us to isolate the spiky components
from the spectrum. As a result, we are able to obtain a more accurate
approximation using fewer Chebyshev moments. Figure \ref{fig:motif_filt}
demonstrates the improvement on the approximation as we procedurally filter out
motifs at $0$, $-1/3$, $-1/2$, and $-1/4$. The eigenvalue $-1/m$ can be
generated by an edge attached to the graph through $m-1$ nodes, similar to
motif (\ref{fig:motif_minushalf}).
\begin{figure}
\begin{subfigure}{0.235\textwidth}
    \centering  
    \captionsetup{justification=centering}
    \includegraphics[width=\textwidth,trim = .3cm 0.5cm 3.5cm 1.3cm,clip]{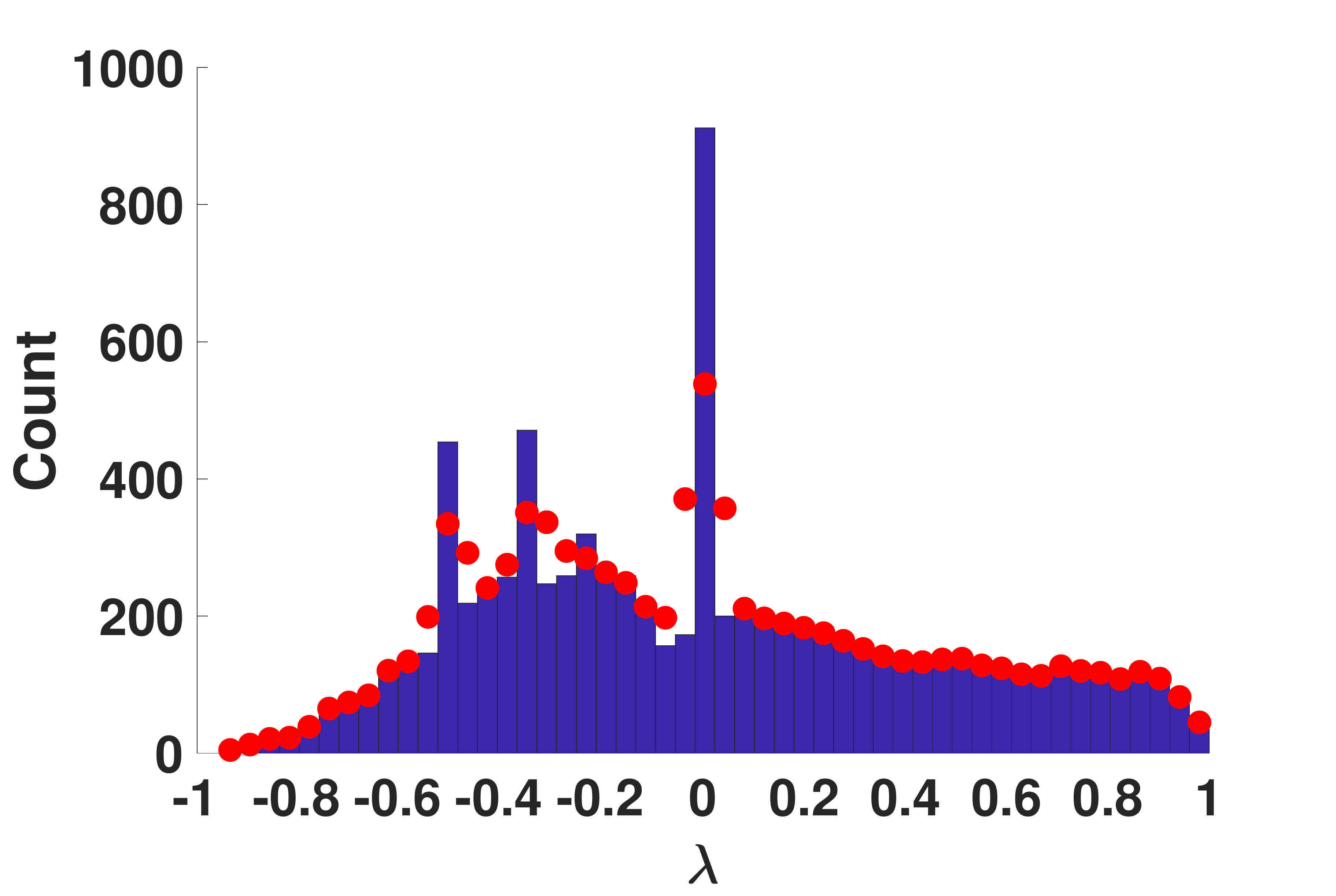}
    \caption{No Filter}
    \label{fig:hepth_0filt}
  \end{subfigure}
  \begin{subfigure}{0.235\textwidth}
    \centering
    \captionsetup{justification=centering}
    \includegraphics[width=\textwidth,trim = .3cm 0.5cm 3.5cm 1.3cm,clip]{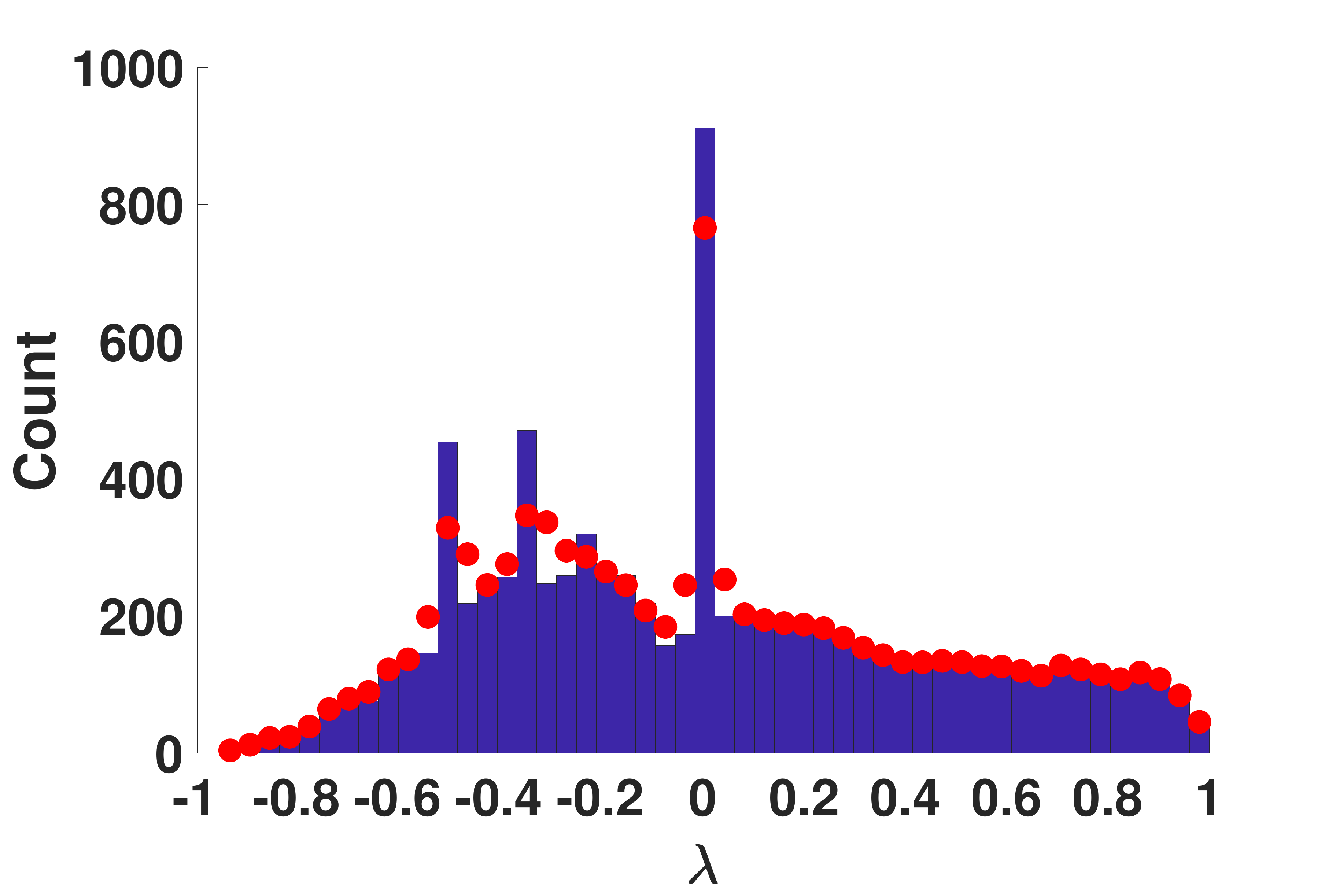}
    \caption{Filter at $\lambda=0$}
    \label{fig:hepth_1filt}
  \end{subfigure}
  \begin{subfigure}{0.235\textwidth}
    \centering
    \captionsetup{justification=centering}
    \includegraphics[width=\textwidth,trim = .3cm 0.5cm 3.5cm 1.3cm,clip]{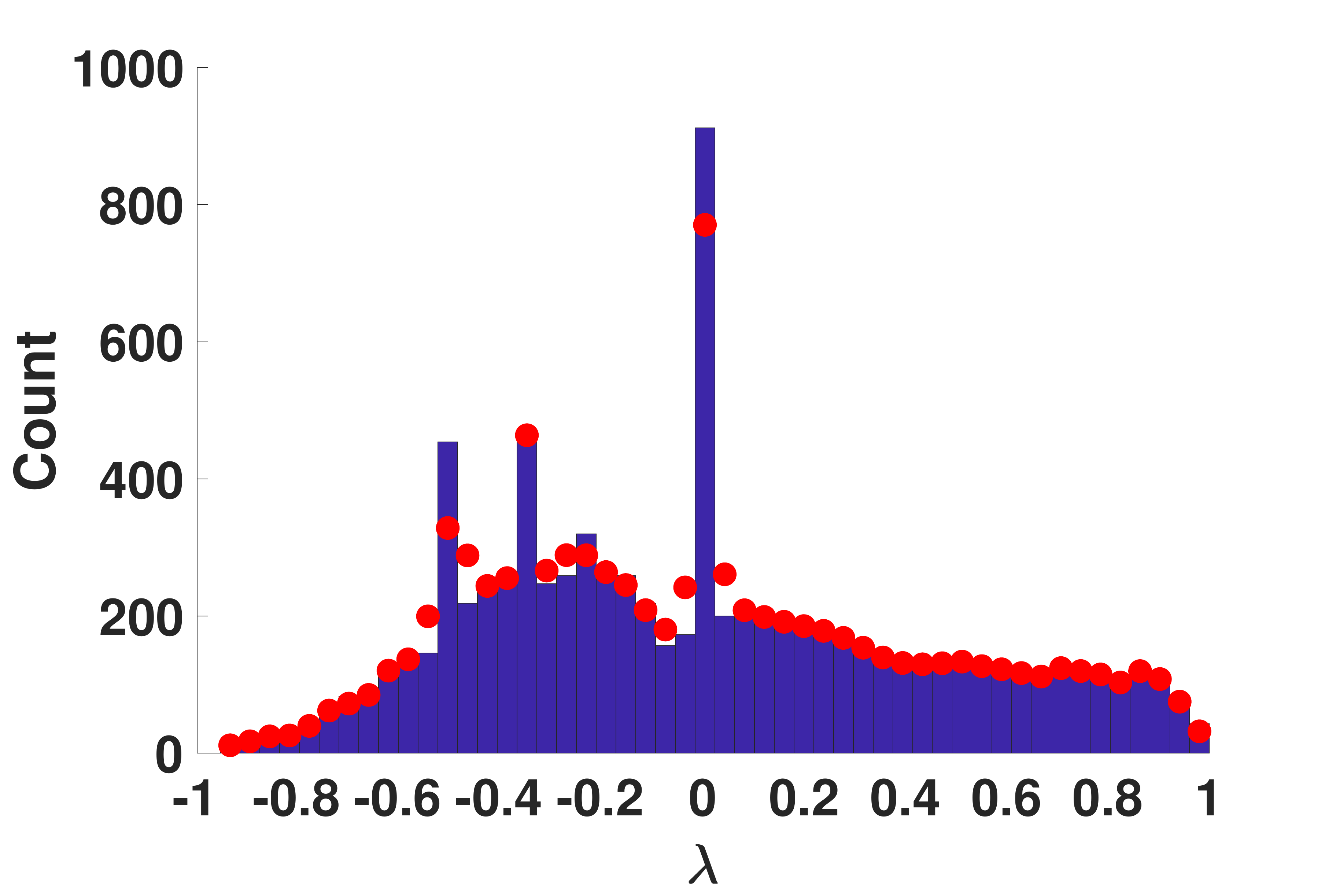}
    \caption{Filter at $\lambda=-1/3$}
    \label{fig:hepth_2filt}
  \end{subfigure}
  \begin{subfigure}{0.235\textwidth}
    \centering
    \captionsetup{justification=centering}
    \includegraphics[width=\textwidth,trim = .3cm 0.5cm 3.5cm 1.3cm,clip]{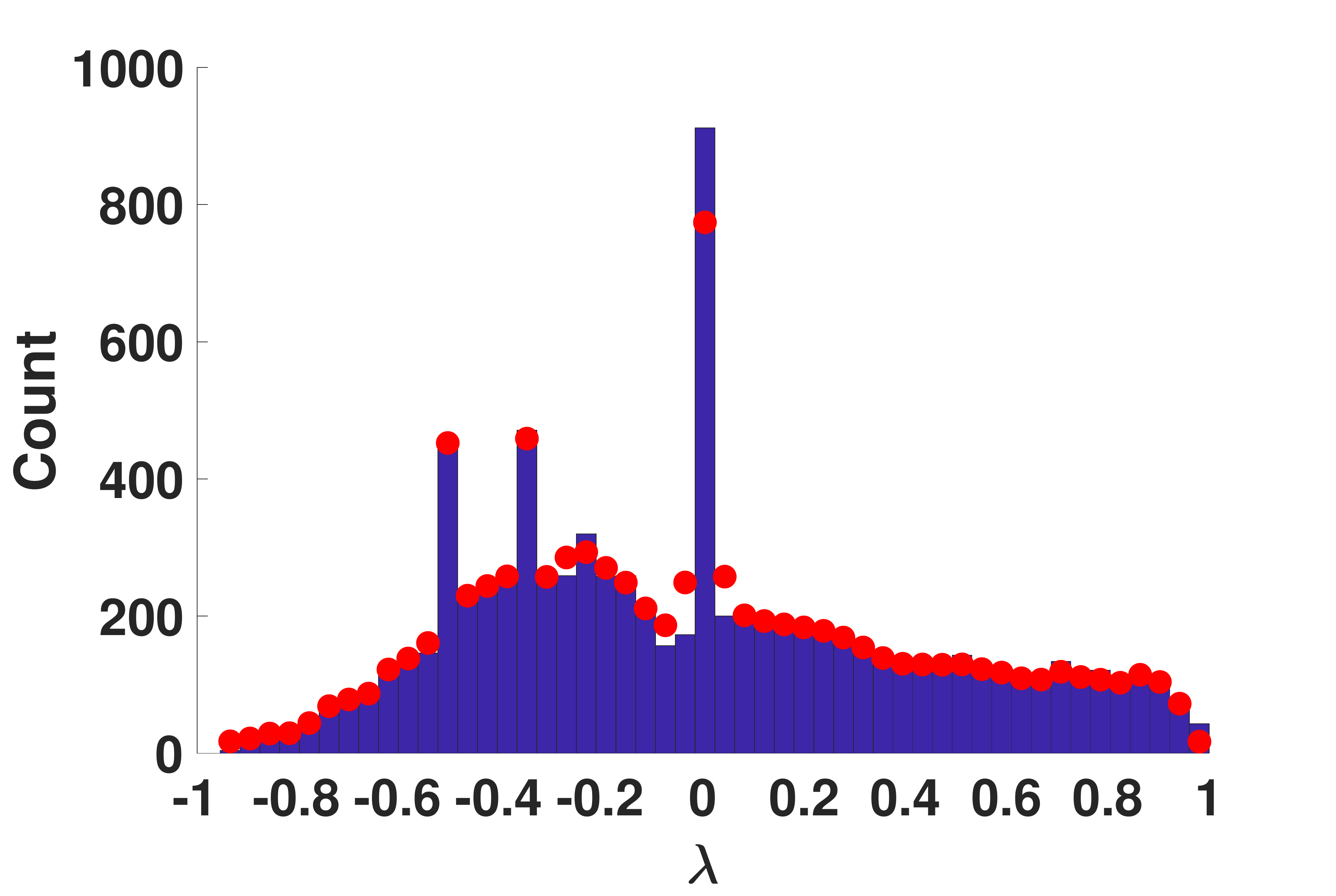}
    \caption{Filter at $\lambda=-1/2$}
    \label{fig:hepth_3filt}
  \end{subfigure}
  \begin{subfigure}{0.235\textwidth}
    \centering
    \captionsetup{justification=centering}
    \includegraphics[width=\textwidth,trim = .3cm 0.5cm 3.5cm 1.3cm,clip]{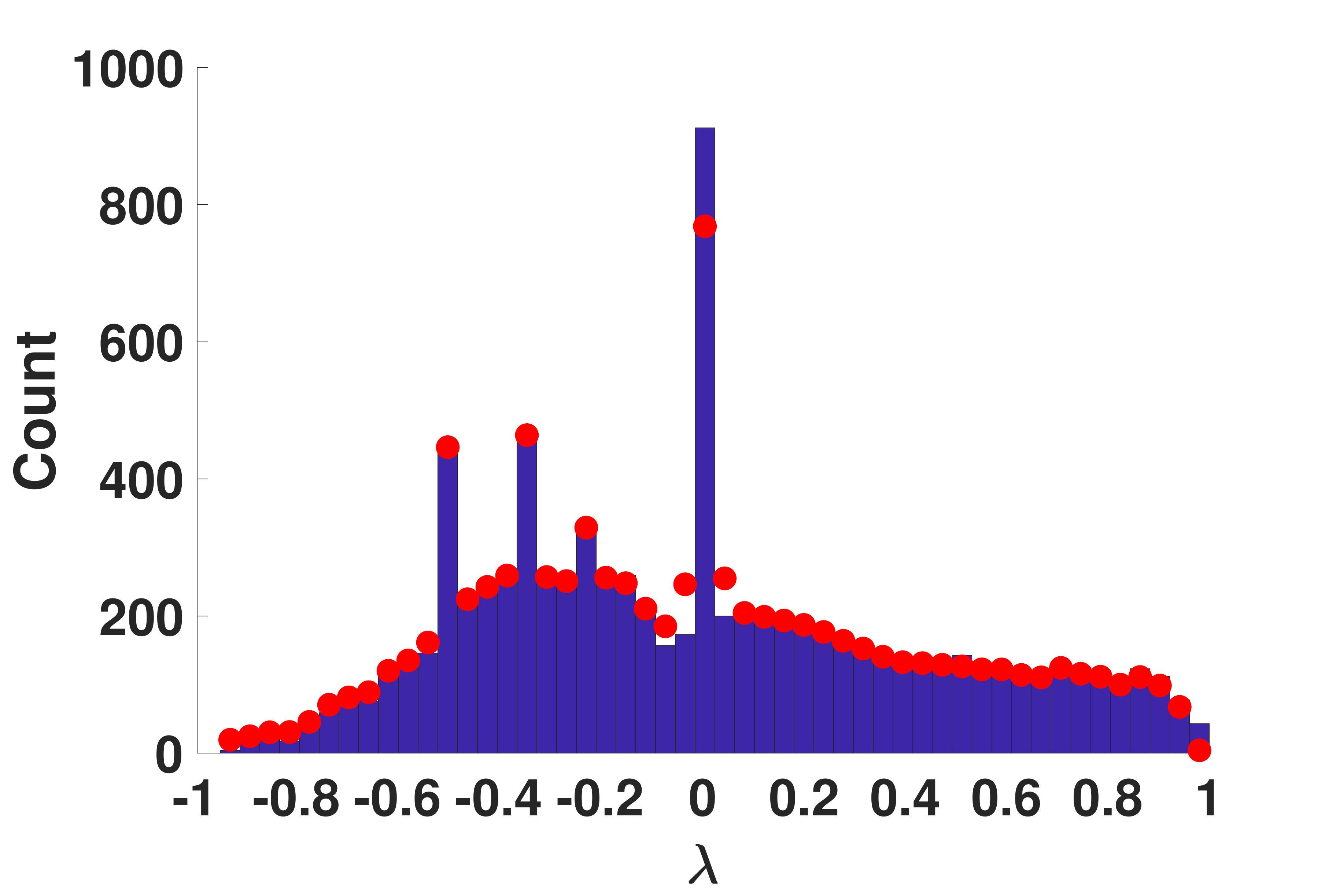}
    \caption{Filter at $\lambda=-1/4$}
    \label{fig:hepth_4filt}
  \end{subfigure}
  \begin{subfigure}{0.235\textwidth}
    \centering
    \captionsetup{justification=centering}
    \includegraphics[width=\textwidth,trim = .3cm 0cm 1.5cm 1.3cm,clip]{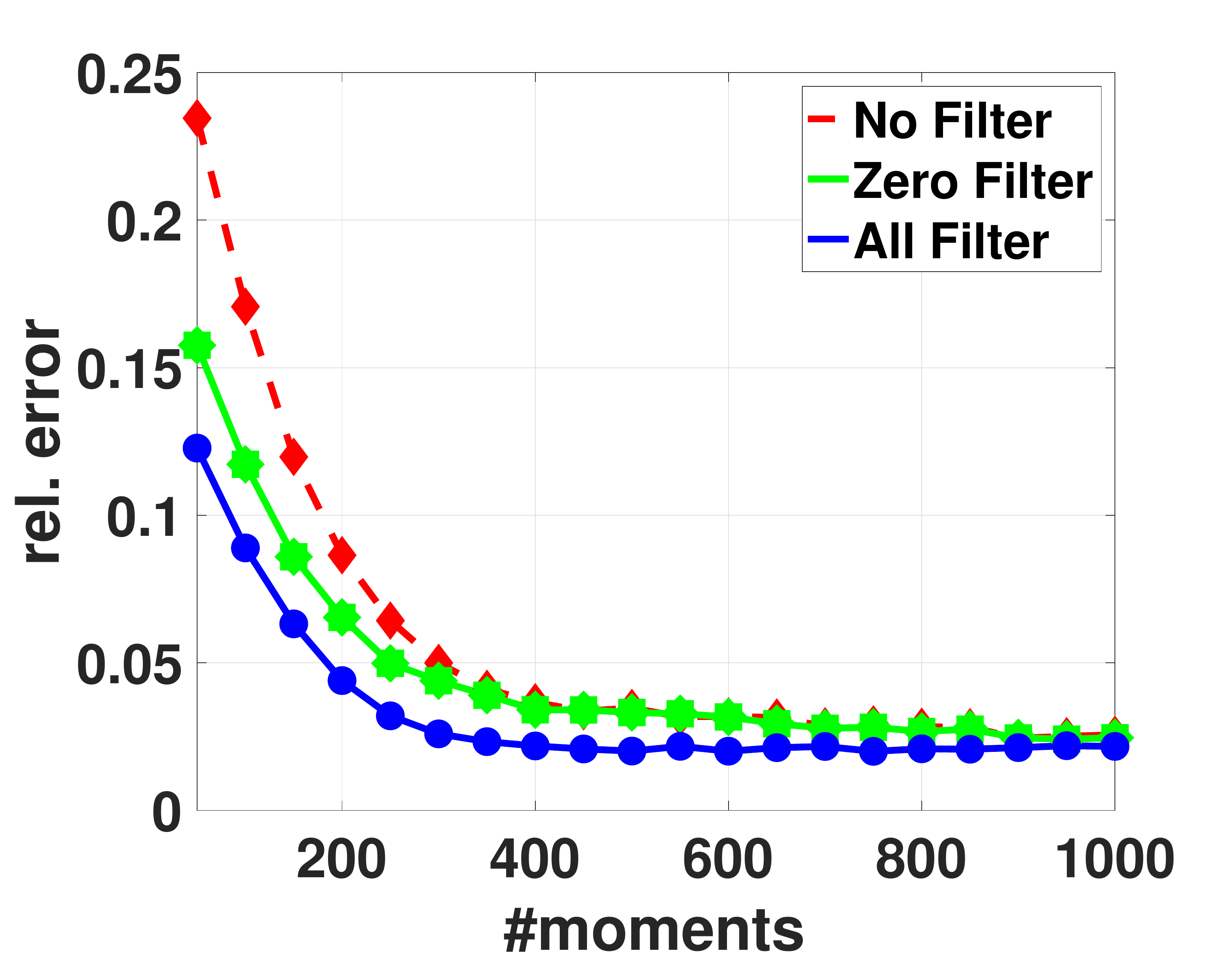}
    \caption{Relative Error}
    \label{fig:hepth_error}
  \end{subfigure}
  \caption{The improvement in accuracy of the spectral histogram
  approximation on the normalized adjacency matrix for the High Energy Physics
  Theory (HepTh) Collaboration Network, as we sweep through spectrum and
  filter out motifs. The graph has $8638$ nodes and $24816$ edges. Blue bars
  are the real spectrum, and red points are the approximated heights. 
  (\ref{fig:hepth_0filt}-\ref{fig:hepth_4filt}) use $100$ moments and $20$ 
  probe vectors. (\ref{fig:hepth_error}) shows the relative $L_1$ error of the
  spectral histogram when using no filter, filter at $\lambda=0$, and all
  filters.}
  \label{fig:motif_filt}
\end{figure}

\section{Error Analysis}
\label{sec:analysis}

\subsection{KPM Approximation Error}
This section provides an error bound for our regularized DOS approximation
$K_\sigma\ast\mu$. We will start with the following theorem.
\begin{theorem}[Jackson's Theorem~\cite{jackson1911genauigkeit}]\label{thm:Jackson_damping}
If $f:[-1,1]\rightarrow \bbR$ is Lipschitz continuous with constant $L$, its
best degree $M$ polynomial approximation $\widehat{f}^M$ has an $L_\infty$ 
error of at most $6L/M$. The approximation can be constructed as
\begin{equation*}
\widehat{f}^M = \sum_{m=0}^MJ_mc_mT_m(x)
\end{equation*}
where $J_m$ are Jackson smoothing factors and $c_m$ are the Chebyshev
coefficients.
\end{theorem}
We can pick a smooth mollifier $K$ with $\text{Lip}(K)=1$. For any $\nu\in\bbR$
and $\lambda\in [-1,1]$ there exists a degree $M$ polynomial such that 
\begin{equation*}
|K_\sigma(\nu-\lambda)-\widehat{K}^{M}_\sigma(\nu-\lambda)| < \frac{6L}{M\sigma}
\end{equation*}
Define $\widehat{\mu}^M = \sum_{m=0}^MJ_md_m\phi_m$ to be the truncated DOS
series,
\begin{equation*}
\int_{-1}^1 \widehat{f}^M(\lambda)\mu(\lambda)d\lambda = \int_{-1}^1f
(\lambda)\widehat{\mu}^M(\lambda)d\lambda = \sum_{m=0}^MJ_mc_md_m.
\end{equation*}
Therefore,
\begin{align*}
\|K_\sigma\ast(\mu-\widehat{\mu}^M)\|_\infty
= &\max_\nu\left|\int_{-1}^1K_\sigma(\nu-\lambda)(\mu(\lambda)-\widehat{\mu}^M
(\lambda))d\lambda\right|\\
\leq &\max_\nu\int_{-1}^1|K_\sigma(\nu-\lambda)-\widehat{K}^M_\sigma
(\nu-\lambda)|\mu(\lambda)d\lambda\\
\leq &\frac{6L}{M\sigma}.
\end{align*}
Consider $\widetilde{\mu}^M$ to be the degree $M$ approximation from KPM,
\begin{equation*}
\|K_\sigma\ast(\mu-\widetilde{\mu}^M)\|_\infty \leq \|K_\sigma\ast(\mu-
\widehat{\mu}^M)\|_\infty+ \|K_\sigma\|_\infty\|\widehat{\mu}^M - 
\widetilde{\mu}^M\|_1.
\end{equation*}
If we use a probe $z$ with independent standard normal entries for the
trace estimation,
\begin{equation*}
\widetilde{\mu}(\lambda) = \sum_{i=1}^Nw_i^2\delta(\lambda-\lambda_i)
\end{equation*}
where $w=Q^Tz$ is the weight for $z$ in the eigenbasis. Hence
\begin{equation*}
\|\widehat{\mu}^M-\widetilde{\mu}^M\|_1 \leq \sum_{i=1}^N |1-w_i^2|.
\end{equation*}
Finally,
\begin{equation*}
\bbE\left[\|K_\sigma\ast(\mu-\widetilde{\mu}^M)\|\right]\leq \frac{1}
{\sigma}\left(\frac{6L}{M}+\|K\|_\infty\bbE[|1-w_1^2|]\right)
\end{equation*}
If we take $N_z$ independent probe vectors, then
$N_zw_1^2\sim\chi^2(N_z)$, which means the expectation decays asymptotically
like $\sqrt{2 / (\pi N_z)}$.

\subsection{Perturbation Analysis}
In this section, we limit our attention to symmetric graph matrix $H$.
Extracting graph information using DOS, whether as a distribution for functions
on a graph or as a direct feature in the form of spectral moments, requires
stability under small perturbations. In the case of removing/adding a few 
number of nodes/edges, the Cauchy Interlacing Theorem~\cite{magnus1988matrix}
gives a bound on each individual new eigenvalue by the old ones. For example, 
if we remove $r\ll N$ nodes to get a new graph matrix $
\widetilde{H}$, then
\begin{equation}\label{eqn:interlacing}
\lambda_i(H)\leq \lambda_i(\widetilde{H})\leq \lambda_{i+r}(H)\quad\text{for
}\quad i\leq N-r
\end{equation}
However, this bound may not be helpful when the impact of the change is not
reflected by its size. Hence, we provide a theorem that relates the Wasserstein
distance (see equation \ref{eqn:waisserstein}) change and the Frobenius norm of
the perturbation. Without loss of generality, we assume the eigenvalues of $H$
lie in $[-1,1]$ already.

\begin{theorem}
Suppose $\widetilde{H} = H + \delta H$ is the perturbed graph matrix with
spectral density $\widetilde{\mu}$, then
\begin{equation*}
W_1(\mu, \widetilde{\mu}) \leq \|\delta H\|_F
\end{equation*}
\end{theorem}
\begin{proof}
Let $\calL$ be the space of Lipschitz functions with $f(0)=0$.
\begin{align*}\label{eqn:proof1}
W_1(\mu, \widetilde{\mu}) =& \sup_{f\in\calL, \text{Lip}(f)=1} \int f(\lambda)
(\mu(\lambda)-\widetilde{\mu}(\lambda))d\lambda \\
=&\frac{1}{N}\sup_{f\in\calL, \text{Lip}(f)=1}\tr(f(H)-f(\widetilde{H}))\\
\leq & \sup_{f\in\calL, \text{Lip}(f)=1, \|v\|=1} v^T(f(H)-f(\widetilde{H}))v.
\end{align*}
By Theorem 3.8 from Higham~\cite{higham2008functions}, the perturbation on
$f(H)$ is bounded by the Fr\'{e}chet derivative,
\begin{equation*}
\|f(H) - f(\widetilde{H})\|_2\leq \text{Lip}(f)\|\delta H\|_F + o(\|\delta
H\|_F).
\end{equation*}
\end{proof}

\section{Experiments}
\label{sec:experiments}

\begin{figure*}
  \begin{subfigure}{0.19\textwidth}
    \centering  
    \captionsetup{justification=centering}
    \includegraphics[width=\textwidth,trim = .4cm 0.5cm 3.5cm 1.3cm,clip]
    {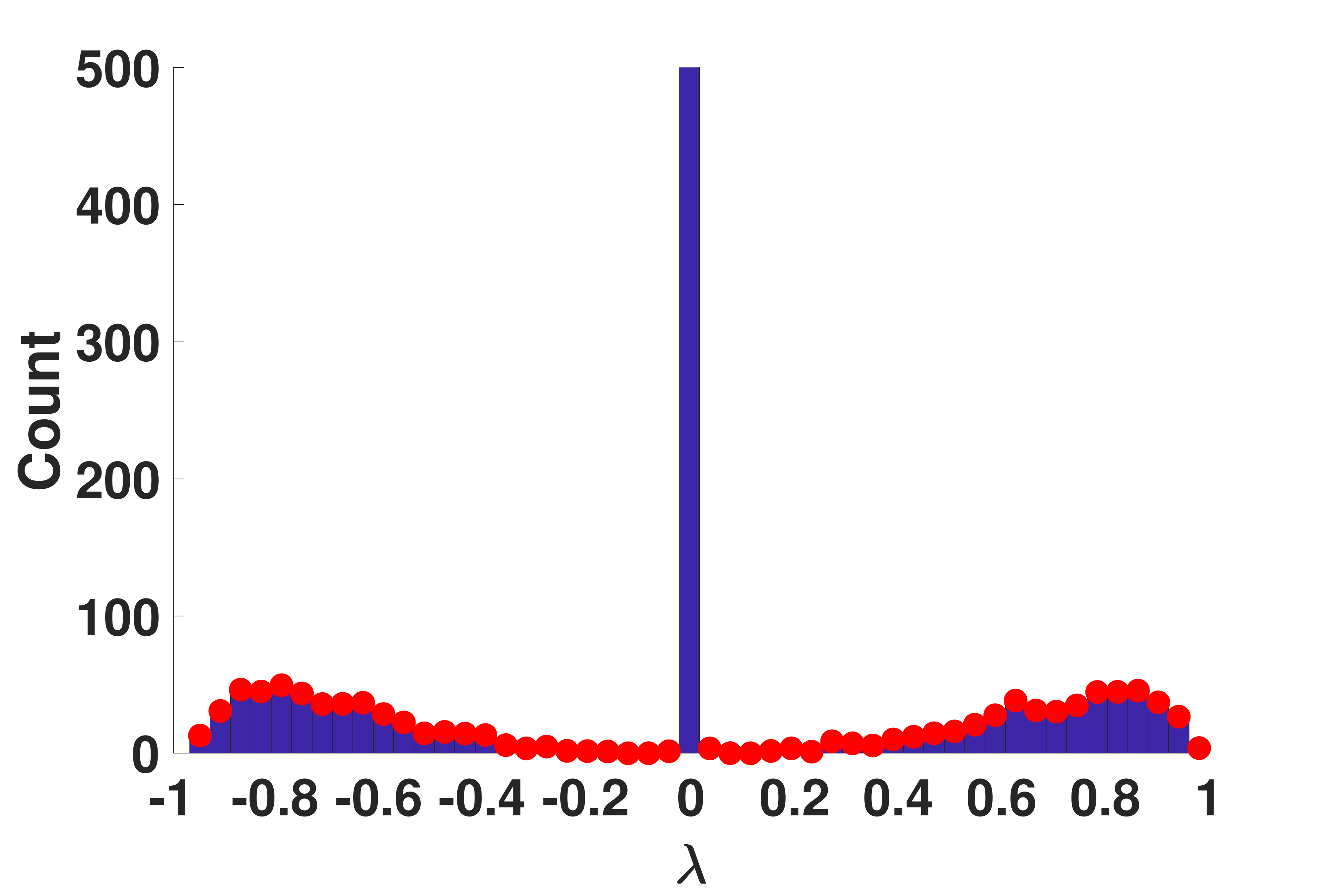}
    \label{fig:erdos_dos}
  \end{subfigure}
  \begin{subfigure}{0.19\textwidth}
    \centering  
    \captionsetup{justification=centering}
    \includegraphics[width=\textwidth,trim = .4cm 0.5cm 3.5cm 1.3cm,clip]
    {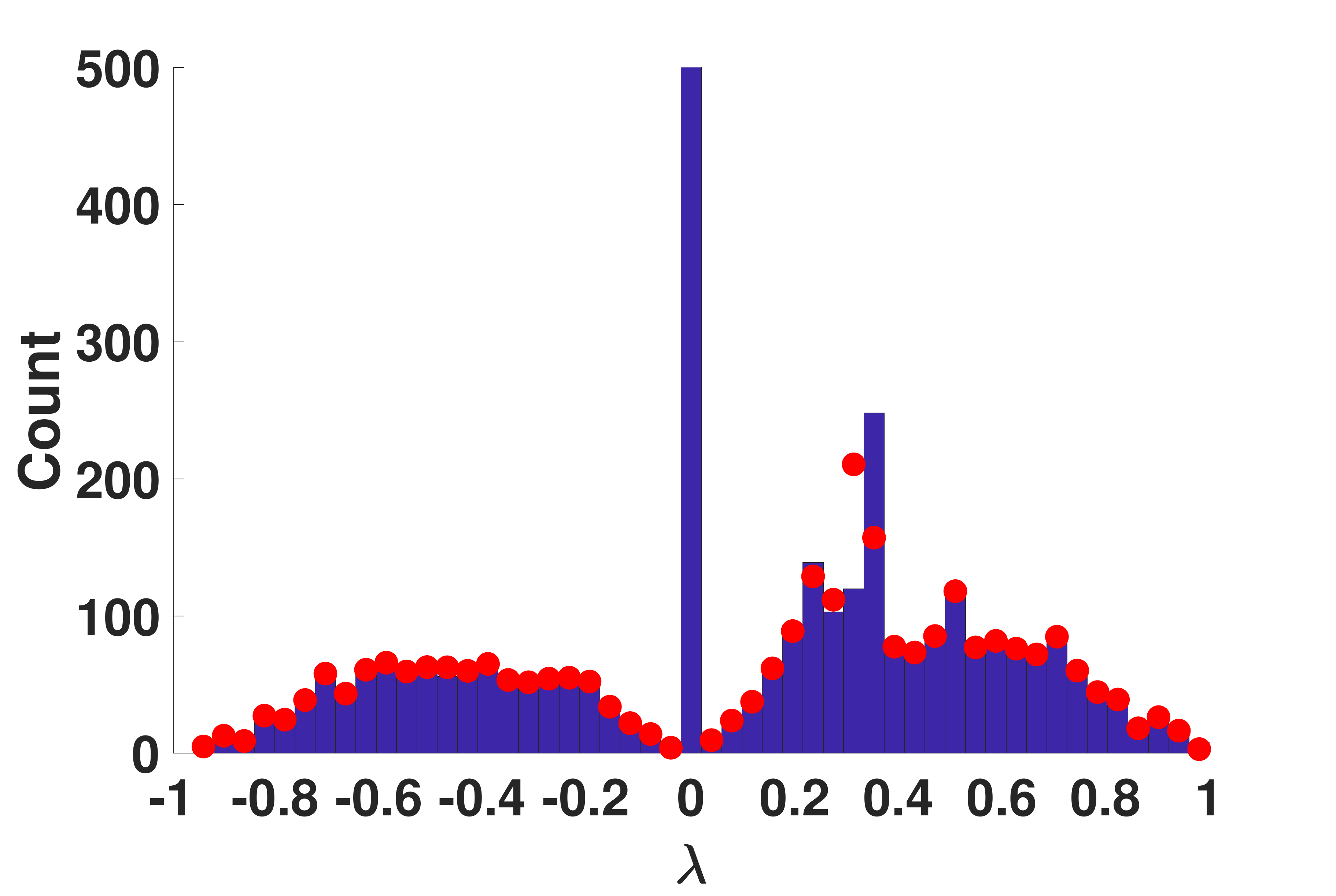}
    \label{fig:as_dos}
  \end{subfigure}
  \begin{subfigure}{0.19\textwidth}
    \centering
    \captionsetup{justification=centering}
    \includegraphics[width=\textwidth,trim = .4cm 0.5cm 3.5cm 1.3cm,clip]
    {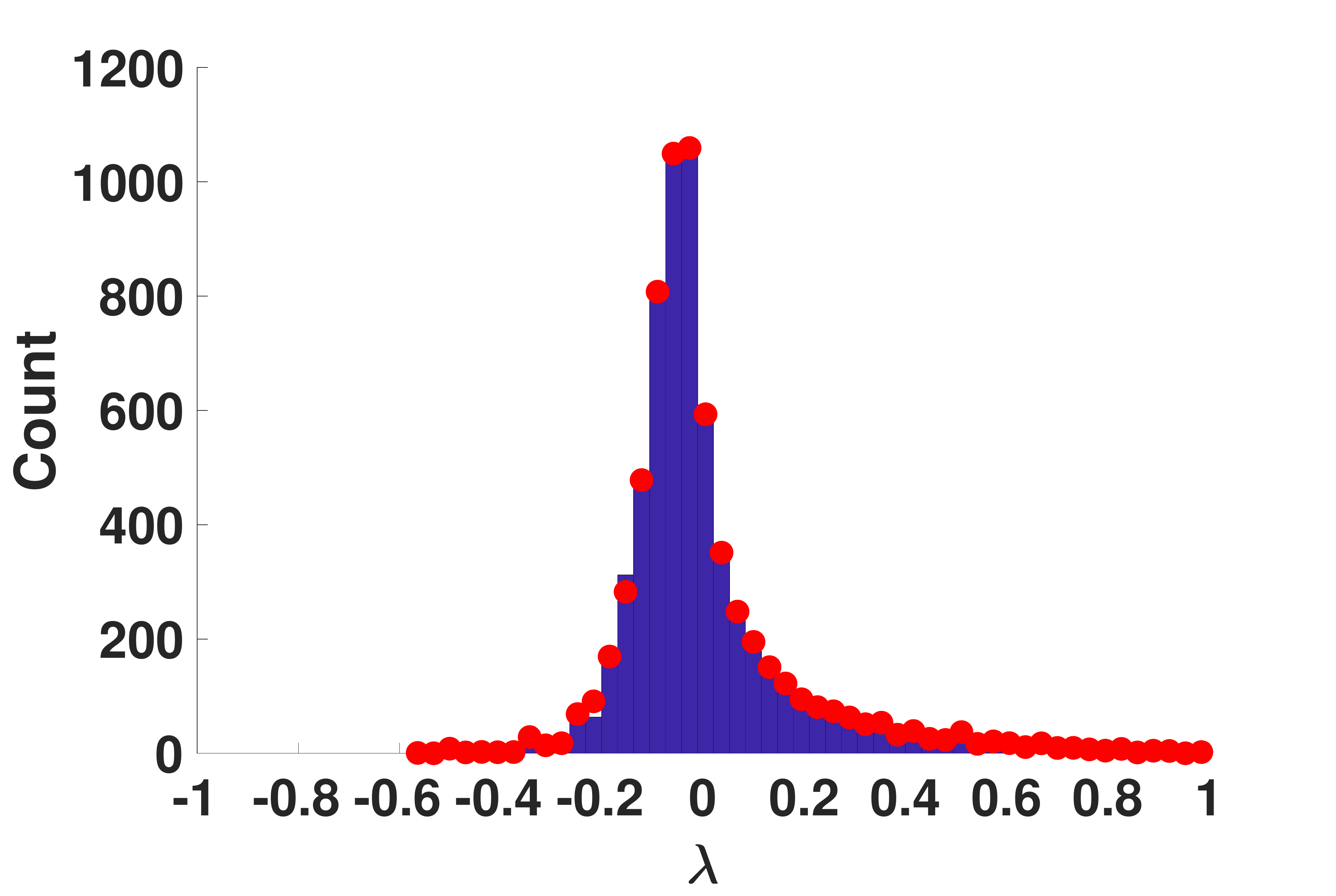}
    \label{fig:marvel_dos}
  \end{subfigure}
  \begin{subfigure}{0.19\textwidth}
    \centering  
    \captionsetup{justification=centering}
    \includegraphics[width=\textwidth,trim = .4cm 0.5cm 3.5cm 1.3cm,clip]
    {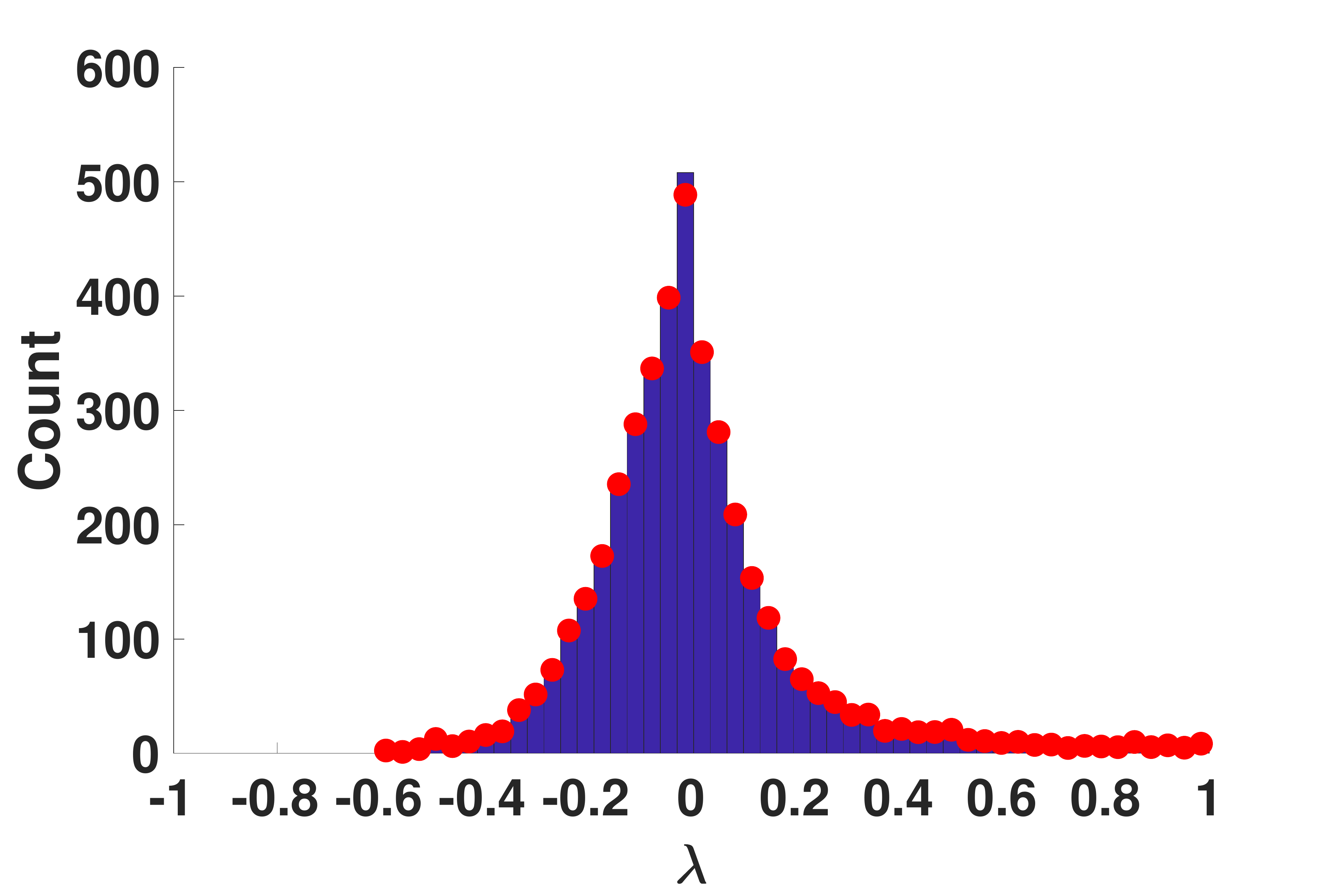}
    \label{fig:facebook_dos}
  \end{subfigure}
  \begin{subfigure}{0.19\textwidth}
    \centering  
    \captionsetup{justification=centering}
    \includegraphics[width=\textwidth,trim = .4cm 0.5cm 3.5cm 1.3cm,clip]
    {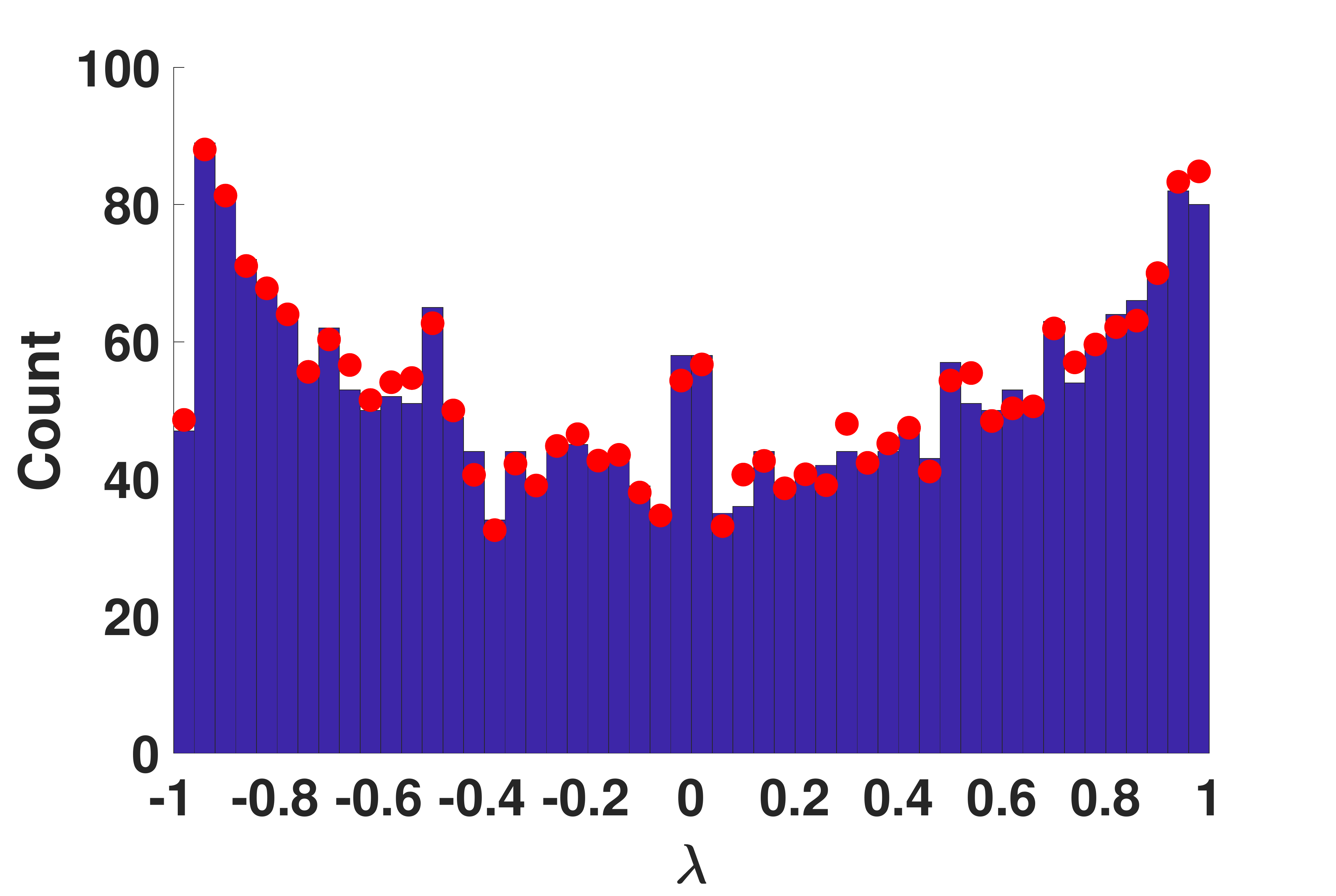}
    \label{fig:minnesota_dos}
  \end{subfigure}
  \begin{subfigure}{0.19\textwidth}
    \centering  
    \captionsetup{justification=centering}
    \includegraphics[width=\textwidth,trim = .4cm 0.5cm 3.5cm 1.3cm,clip]
    {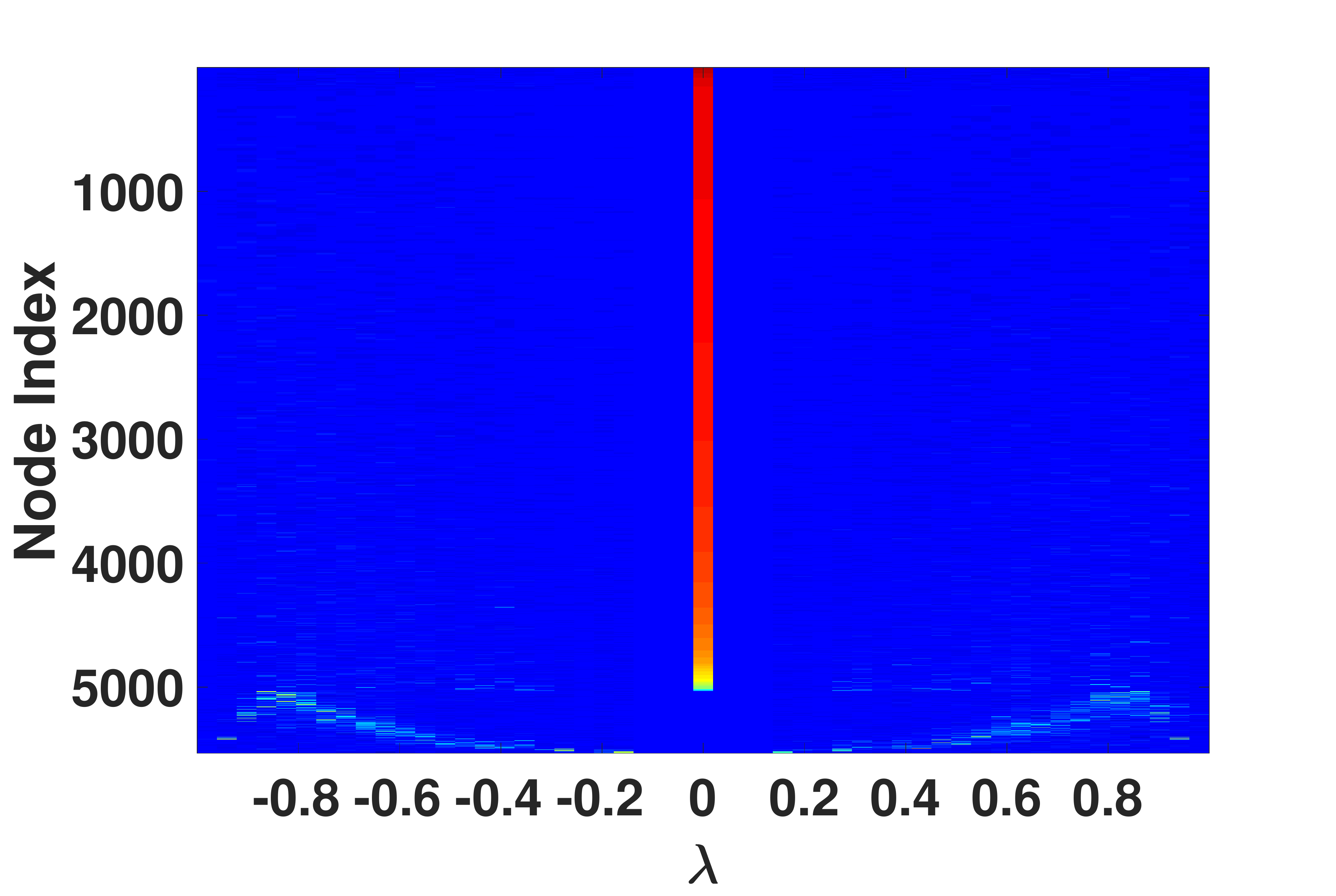}
    \caption{Erd\H{o}s Collaboration Network}
    \label{fig:erdos_ldos}
  \end{subfigure}
  \begin{subfigure}{0.19\textwidth}
    \centering  
    \captionsetup{justification=centering}
    \includegraphics[width=\textwidth,trim = .4cm 0.5cm 3.5cm 1.3cm,clip]
    {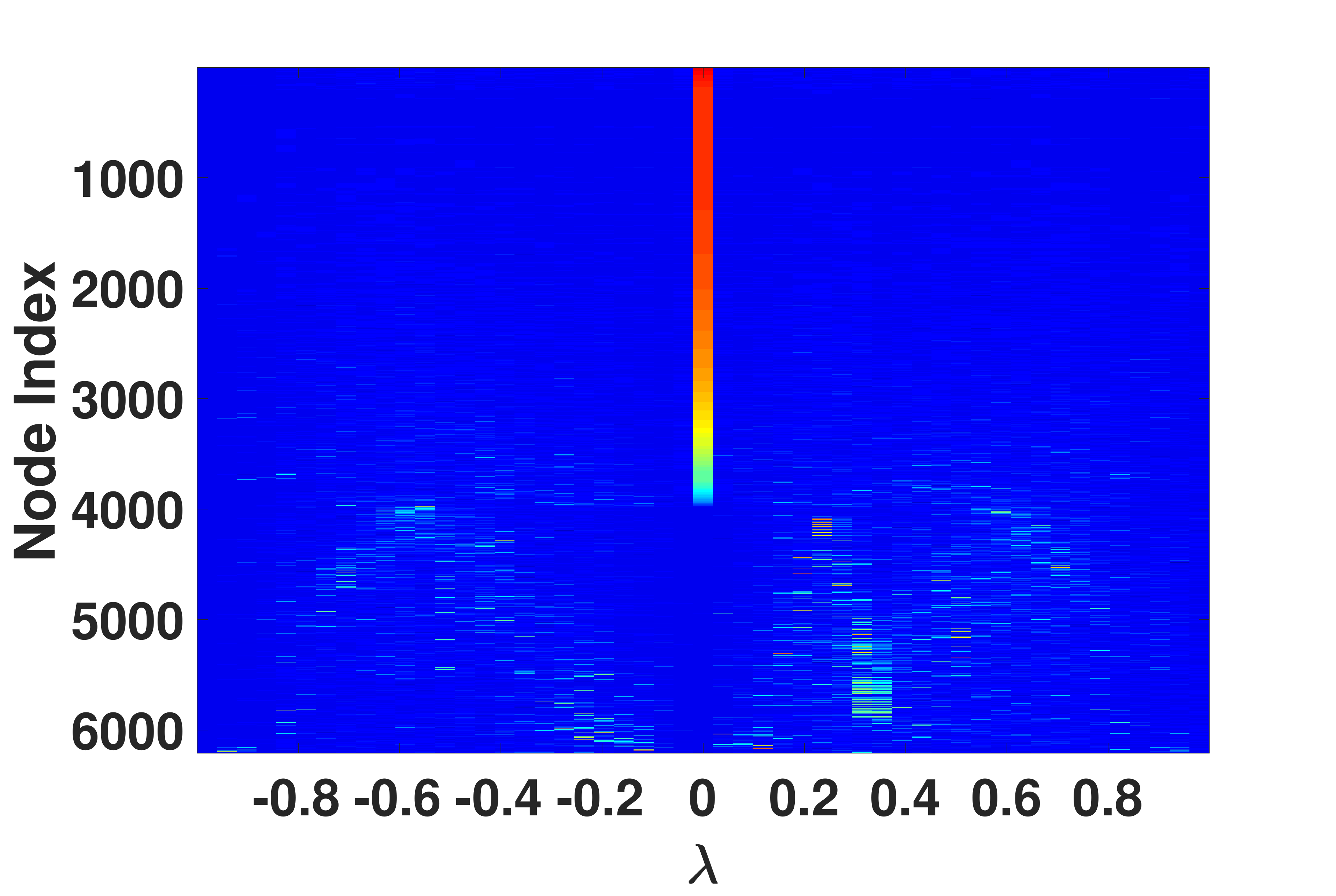}
    \caption{Autonomous System Network (1999)}
    \label{fig:as_ldos}
  \end{subfigure}
  \begin{subfigure}{0.19\textwidth}
    \centering  
    \captionsetup{justification=centering}
    \includegraphics[width=\textwidth,trim = .4cm 0.5cm 3.5cm 1.3cm,clip]
    {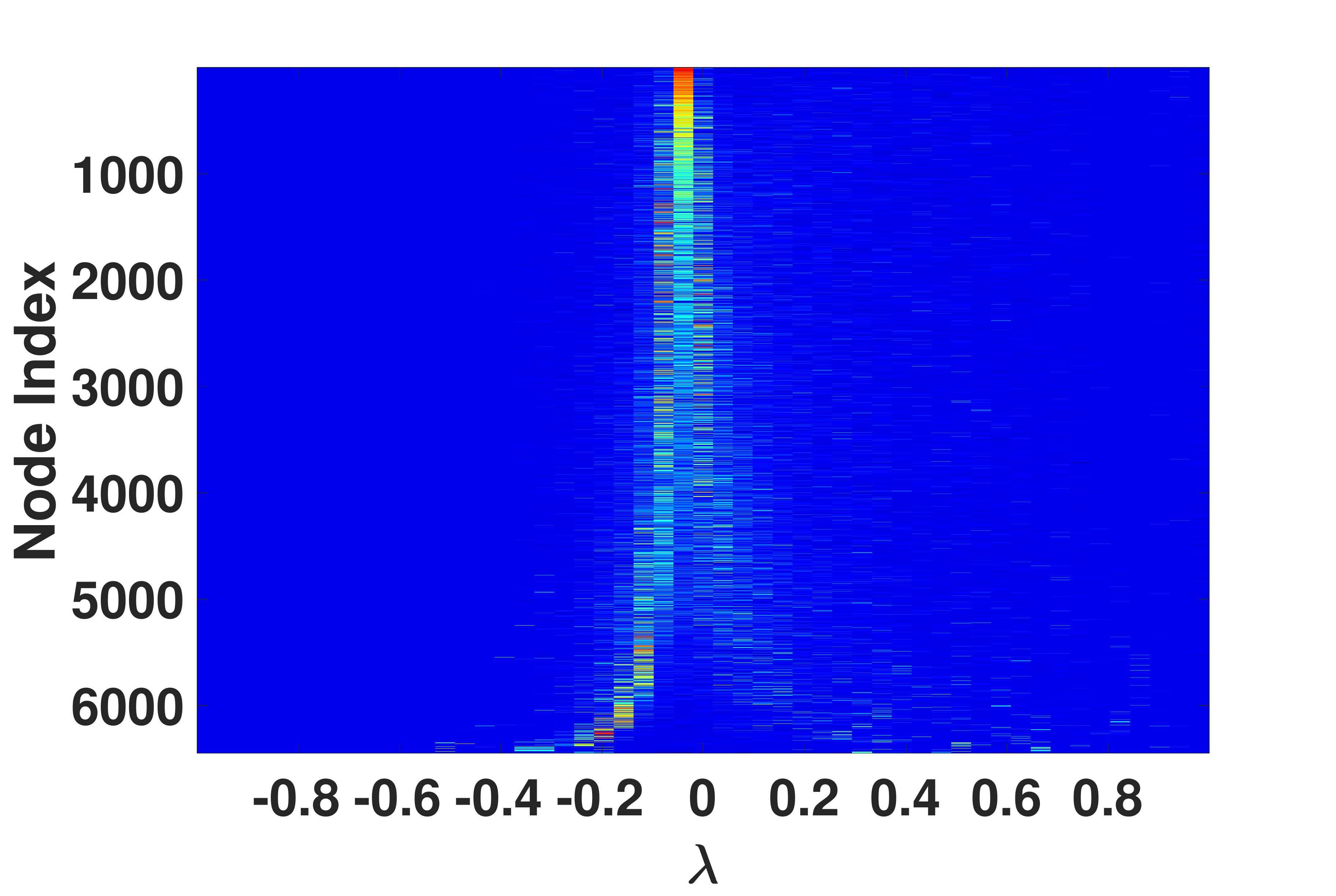}{}
    \caption{Marvel Characters Network}
    \label{fig:marvel_ldos}
  \end{subfigure}
  \begin{subfigure}{0.19\textwidth}
    \centering  
    \captionsetup{justification=centering}
    \includegraphics[width=\textwidth,trim = .4cm 0.5cm 3.5cm 1.3cm,clip]
    {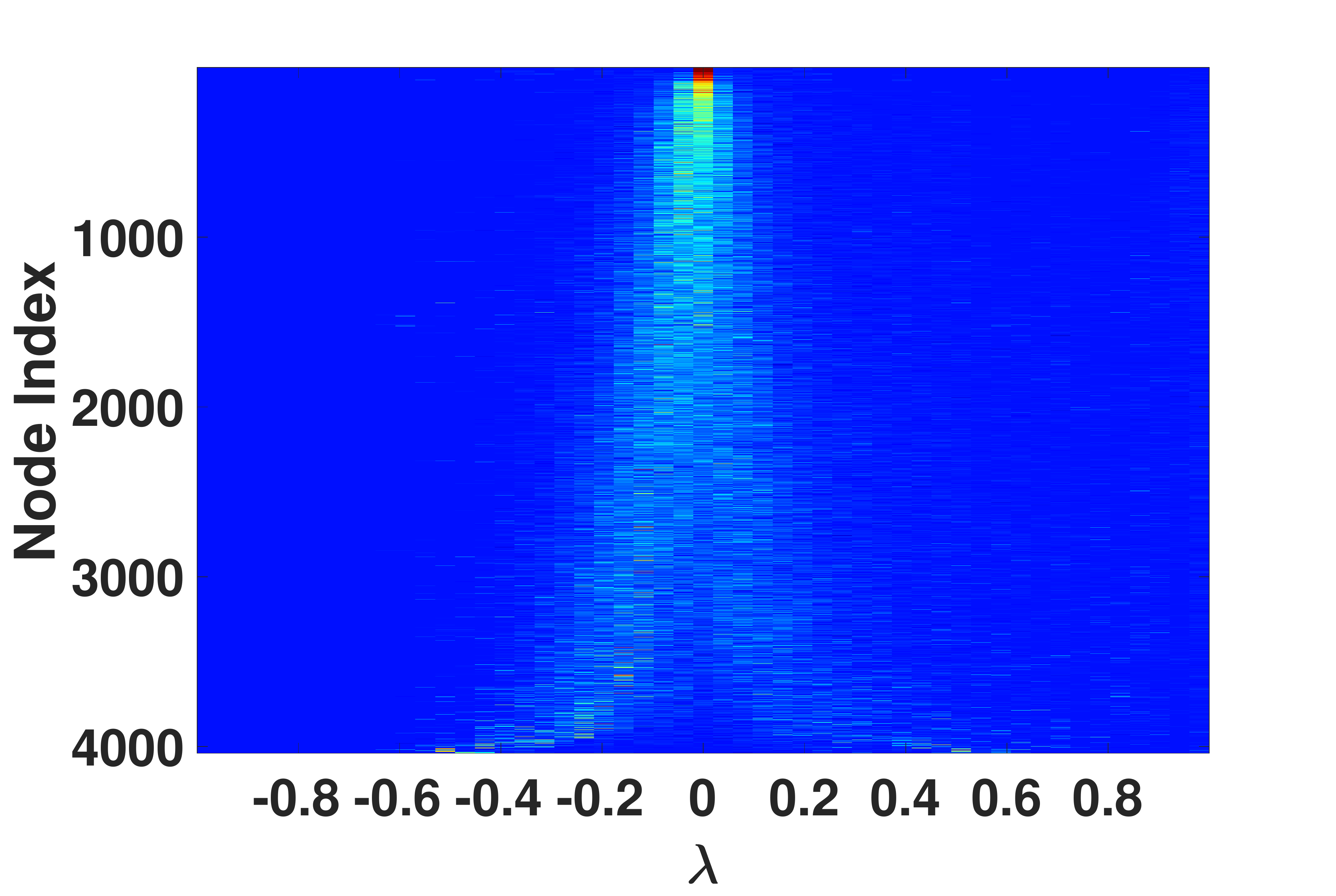}
    \caption{Facebook Ego Networks}
    \label{fig:facebook_ldos}
  \end{subfigure}
  \begin{subfigure}{0.19\textwidth}
    \centering  
    \captionsetup{justification=centering}
    \includegraphics[width=\textwidth,trim = .4cm 0.5cm 3.5cm 1.3cm,clip]
    {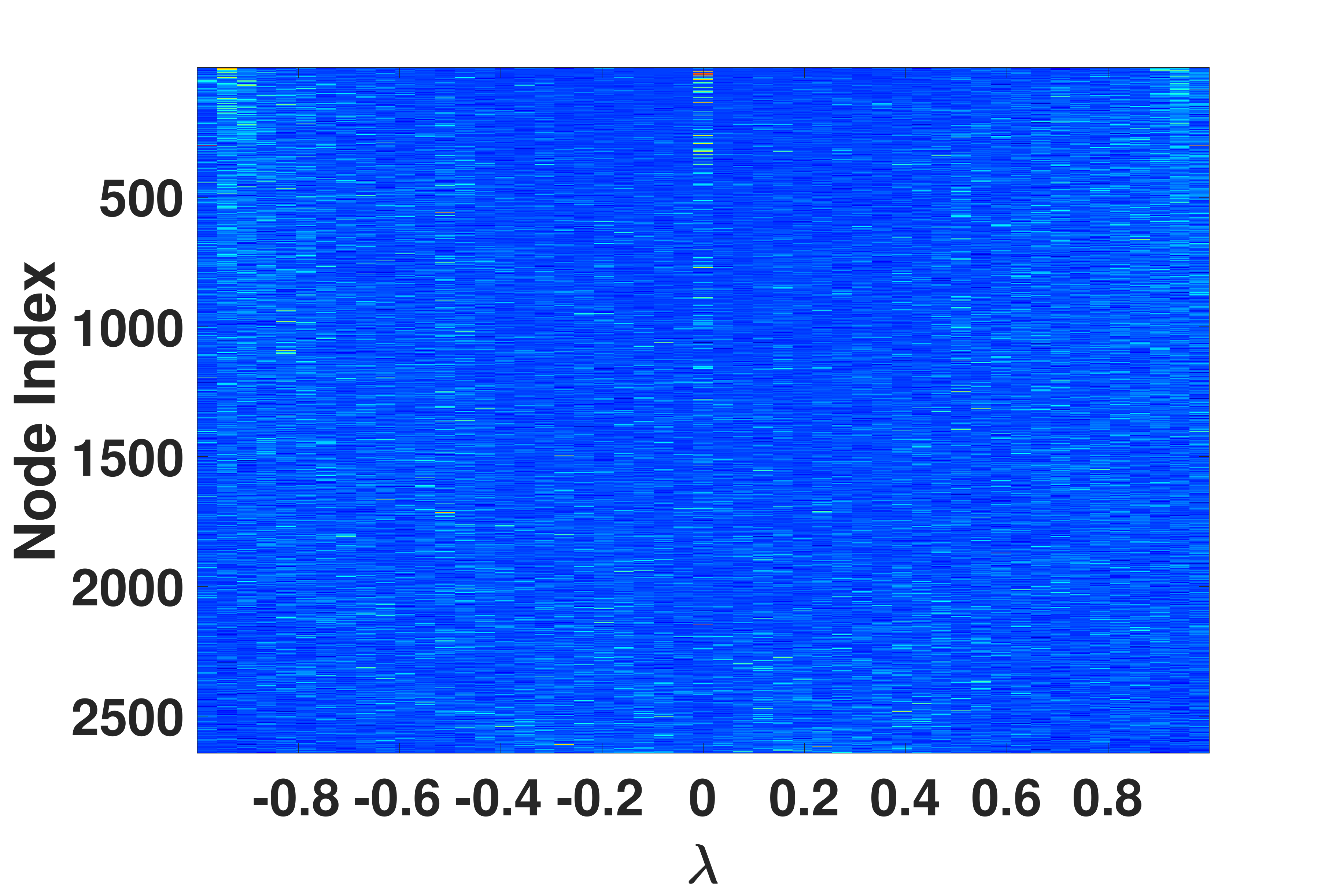}
    \caption{Minnesota Road Network}
    \label{fig:minnesota_ldos}
  \end{subfigure}
  \begin{subfigure}{0.19\textwidth}
    \centering  
    \captionsetup{justification=centering}
    \includegraphics[width=\textwidth,trim = .4cm 0.5cm 3.5cm 1.3cm,clip]
    {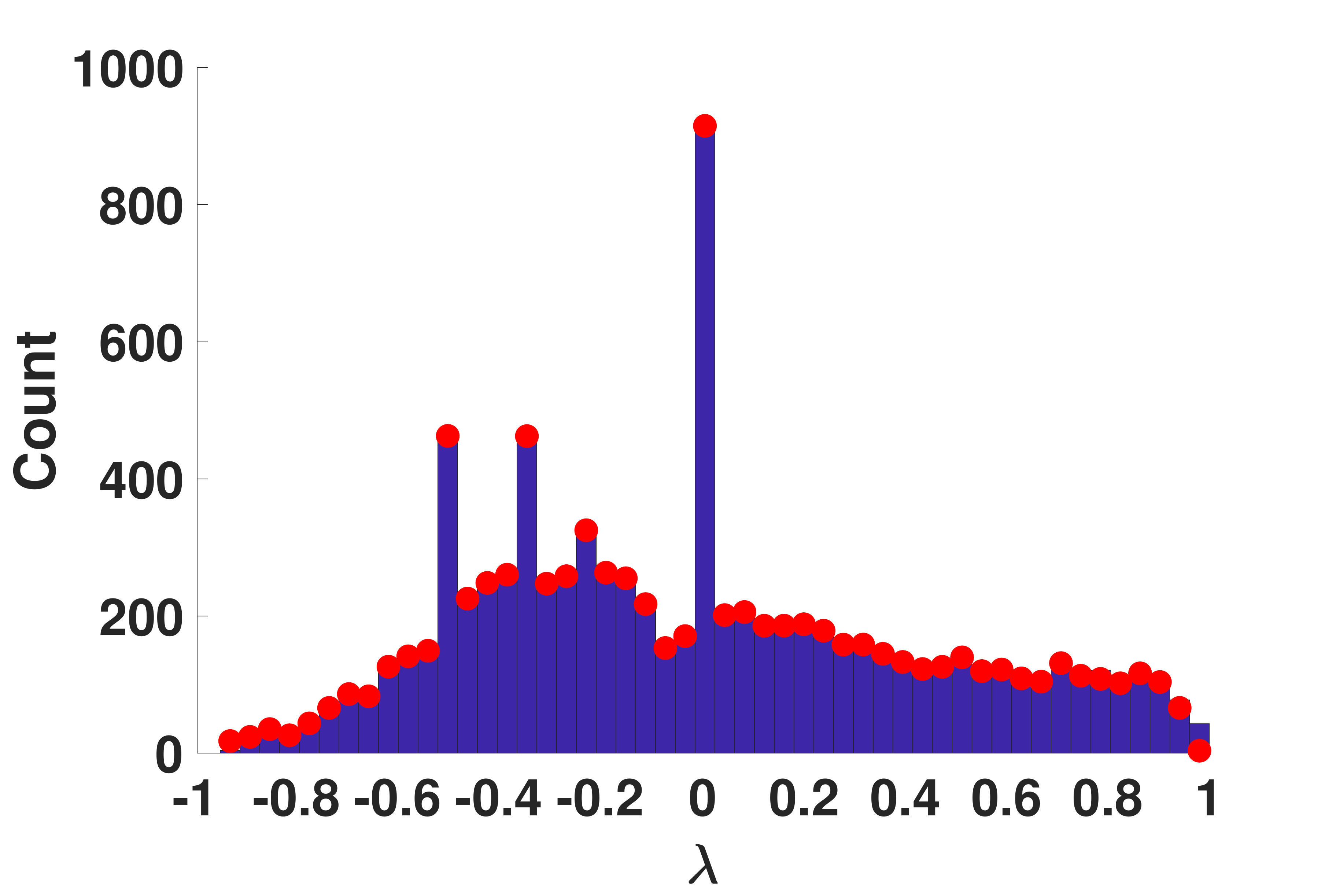}
    \label{fig:hepth_dos}
  \end{subfigure}
  \begin{subfigure}{0.19\textwidth}
    \centering  
    \captionsetup{justification=centering}
    \includegraphics[width=\textwidth,trim = .4cm 0.5cm 3.5cm 1.3cm,clip]
    {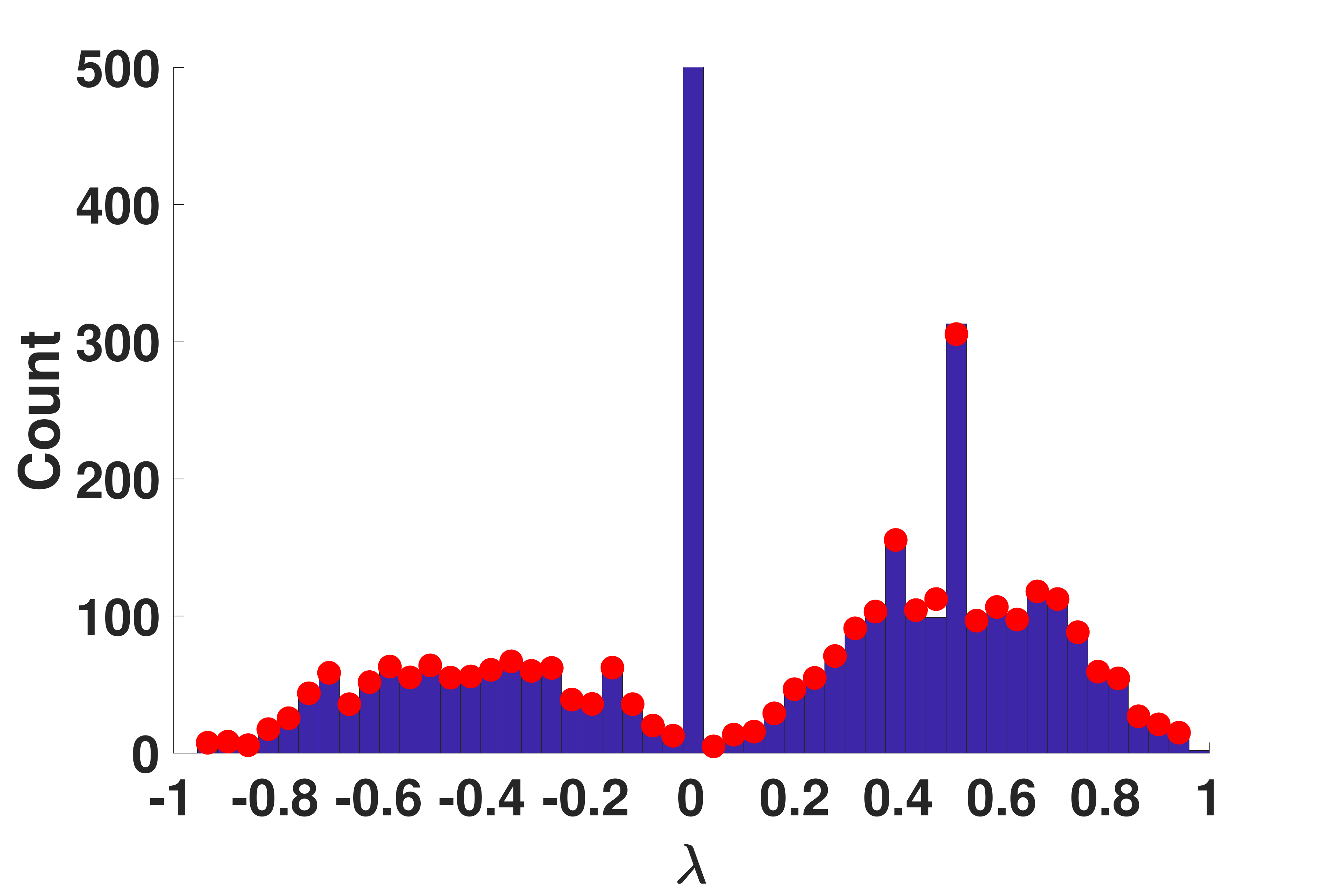}
    \label{fig:as2_dos}
  \end{subfigure}
  \begin{subfigure}{0.19\textwidth}
    \centering
    \captionsetup{justification=centering}
    \includegraphics[width=\textwidth,trim = .4cm 0.5cm 3.5cm 1.3cm,clip]
    {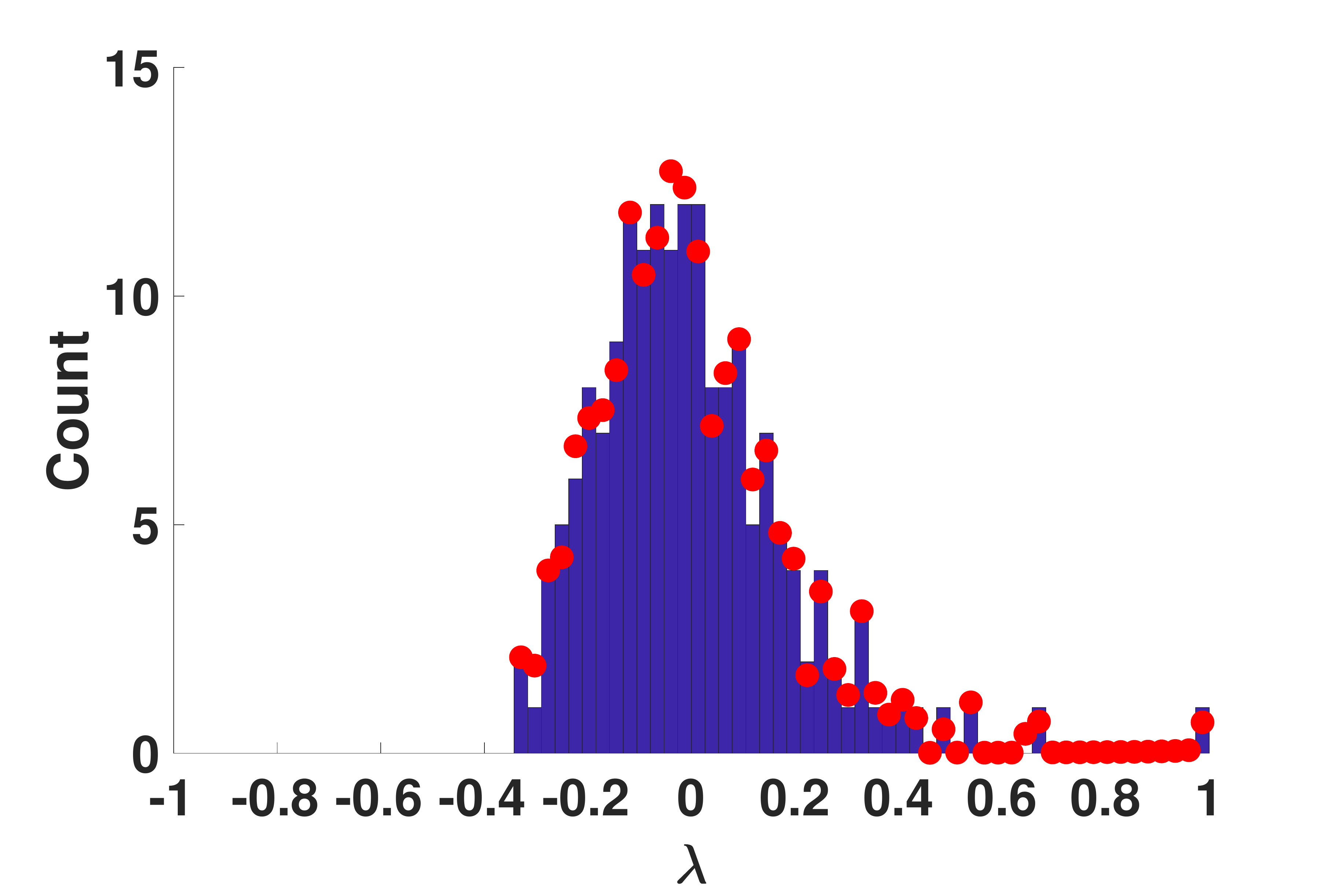}
    \label{fig:hp_dos}
  \end{subfigure}
  \begin{subfigure}{0.19\textwidth}
    \centering  
    \captionsetup{justification=centering}
    \includegraphics[width=\textwidth,trim = .4cm 0.5cm 3.5cm 1.3cm,clip]
    {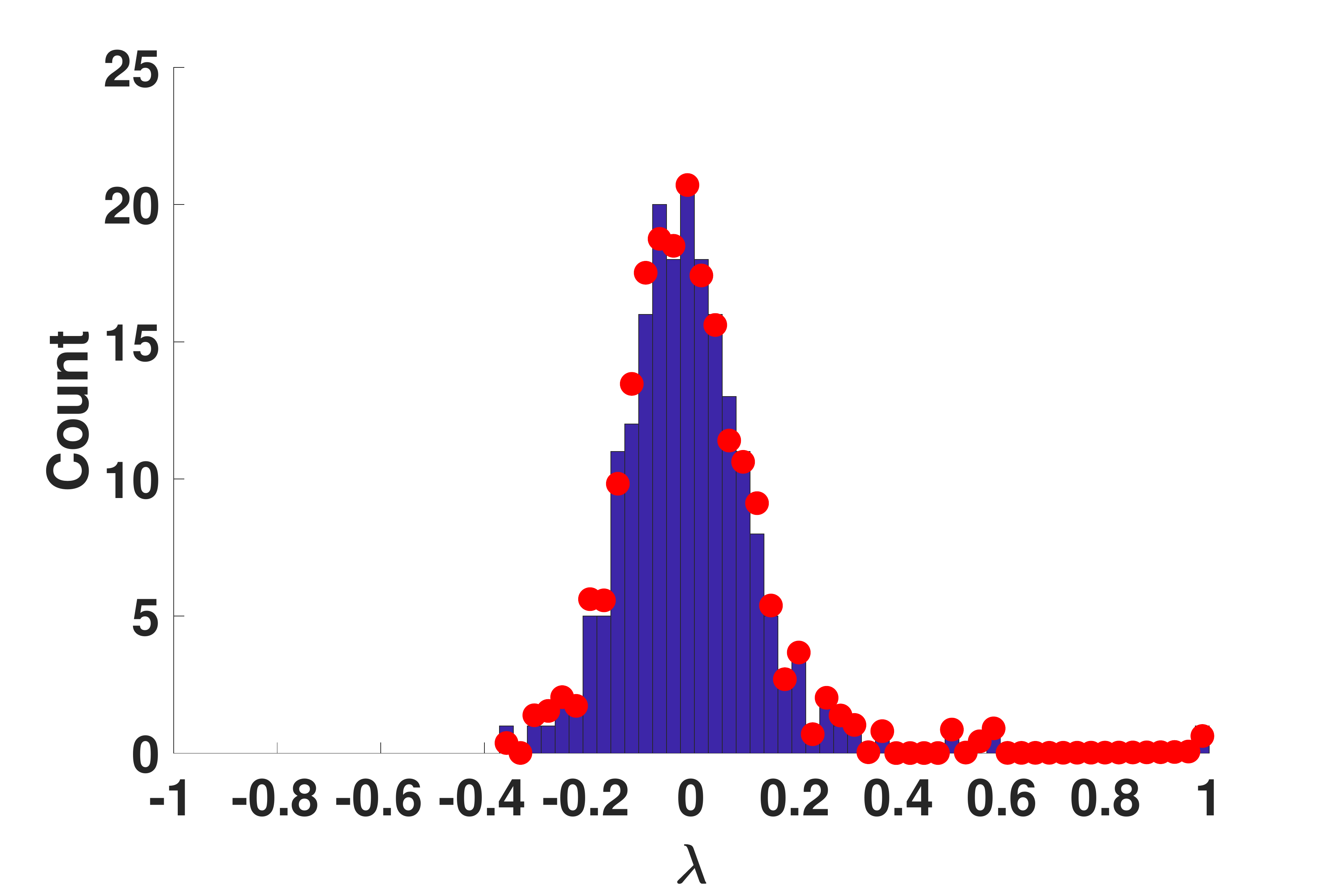}
    \label{fig:twitter_dos}
  \end{subfigure}
  \begin{subfigure}{0.19\textwidth}
    \centering  
    \captionsetup{justification=centering}
    \includegraphics[width=\textwidth,trim = .4cm 0.5cm 3.5cm .3cm,clip]
    {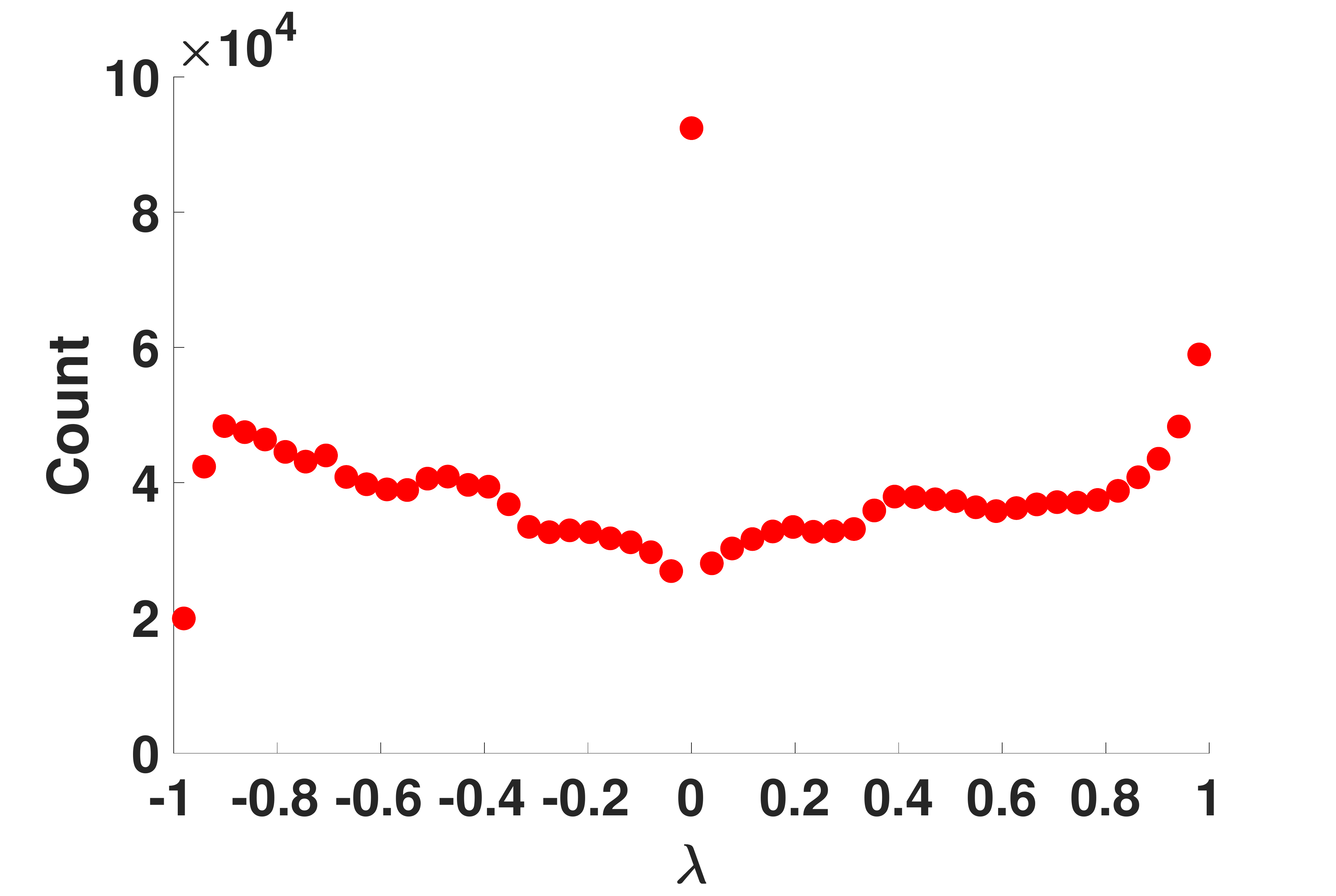}
    \label{fig:roadca_dos}
  \end{subfigure}
  \begin{subfigure}{0.19\textwidth}
    \centering  
    \captionsetup{justification=centering}
    \includegraphics[width=\textwidth,trim = .4cm 0.5cm 3.5cm 1.3cm,clip]
    {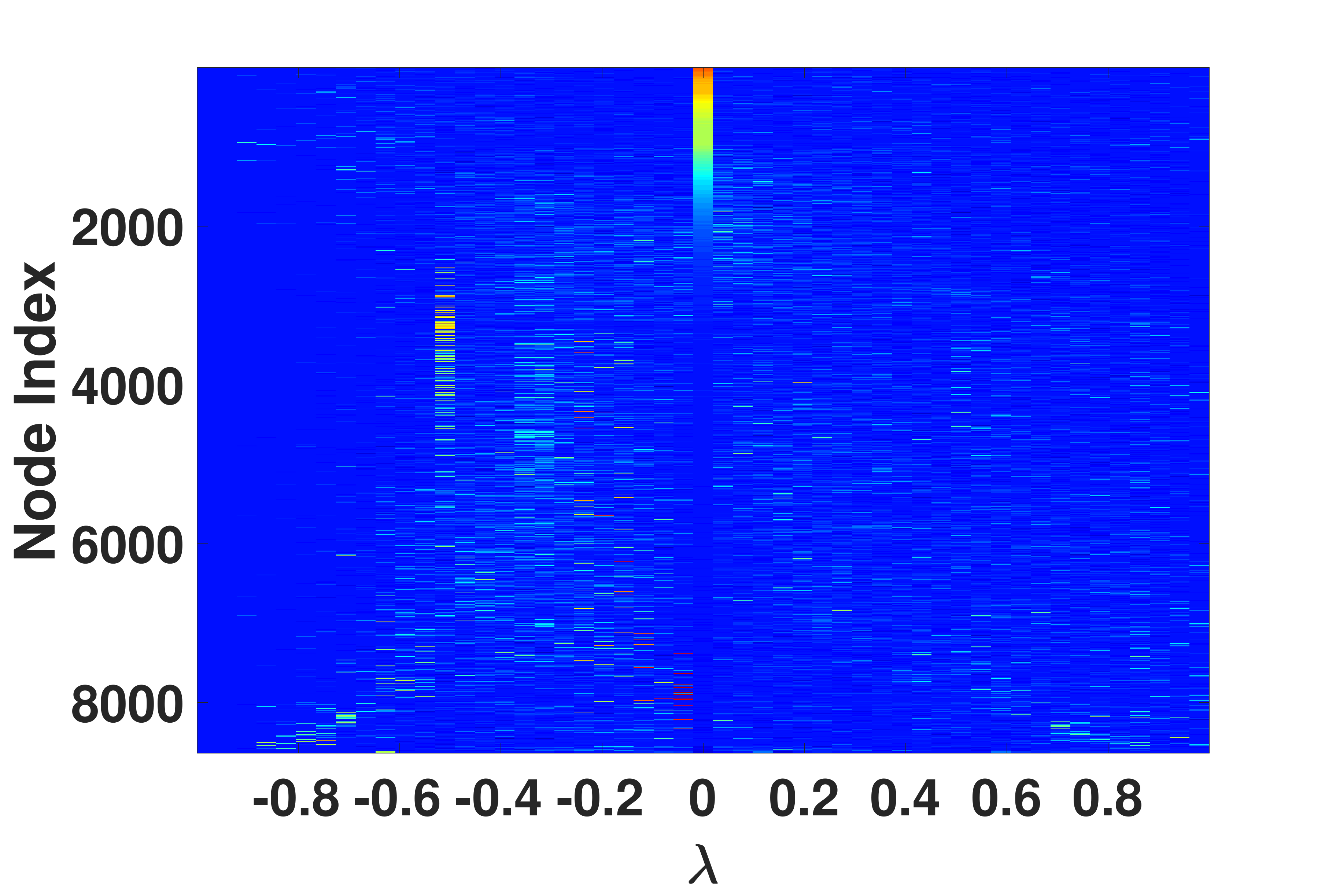}
    \caption{HepTh Collaboration Network}
    \label{fig:hepth_ldos}
  \end{subfigure}
  \begin{subfigure}{0.19\textwidth}
    \centering  
    \captionsetup{justification=centering}
    \includegraphics[width=\textwidth,trim = .4cm 0.5cm 3.5cm 1.3cm,clip]
    {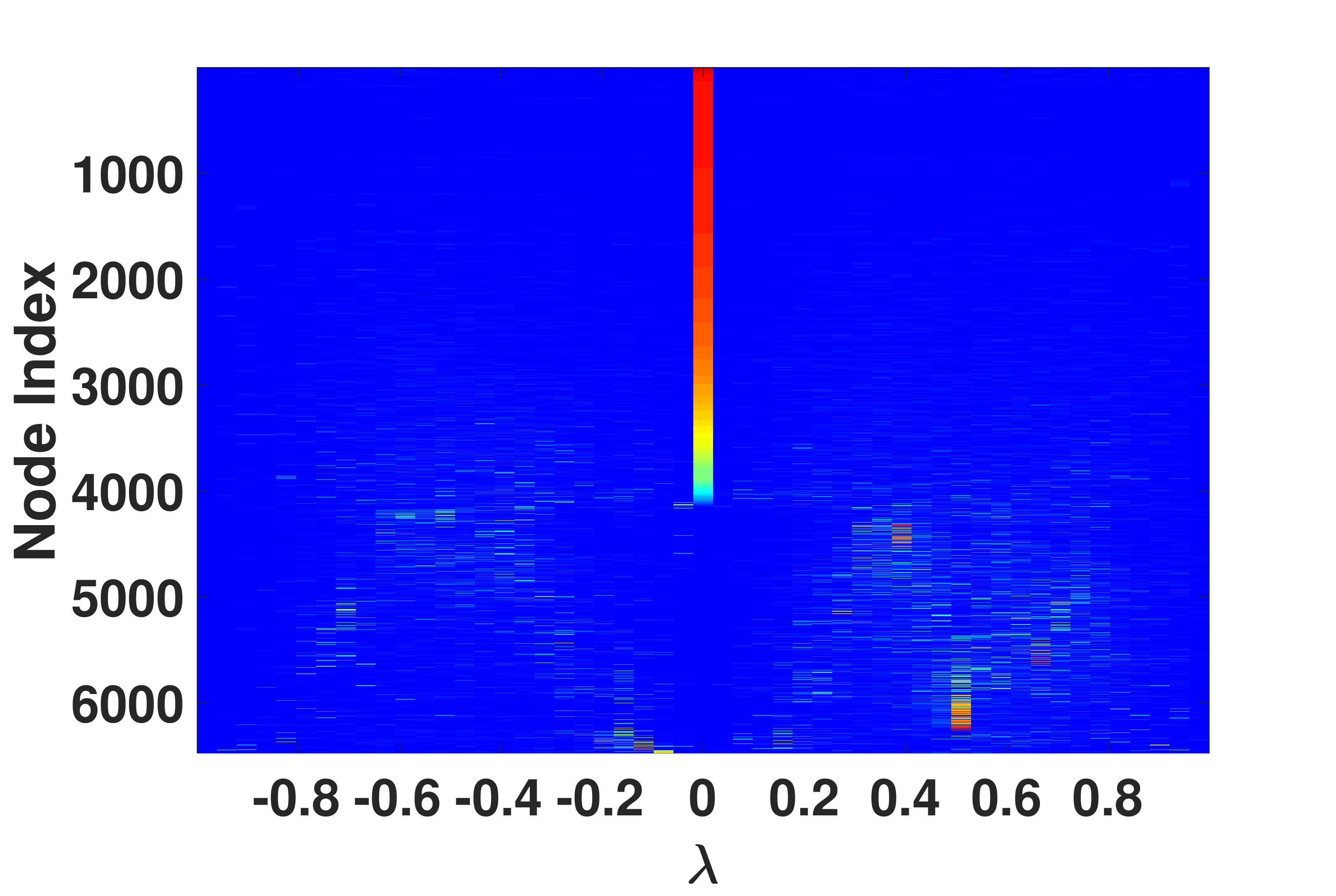}
    \caption{Autonomous System Network (2000)}
    \label{fig:as2_ldos}
  \end{subfigure}
  \begin{subfigure}{0.19\textwidth}
    \centering  
    \captionsetup{justification=centering}
    \includegraphics[width=\textwidth,trim = .4cm 0.5cm 3.5cm 1.3cm,clip]
    {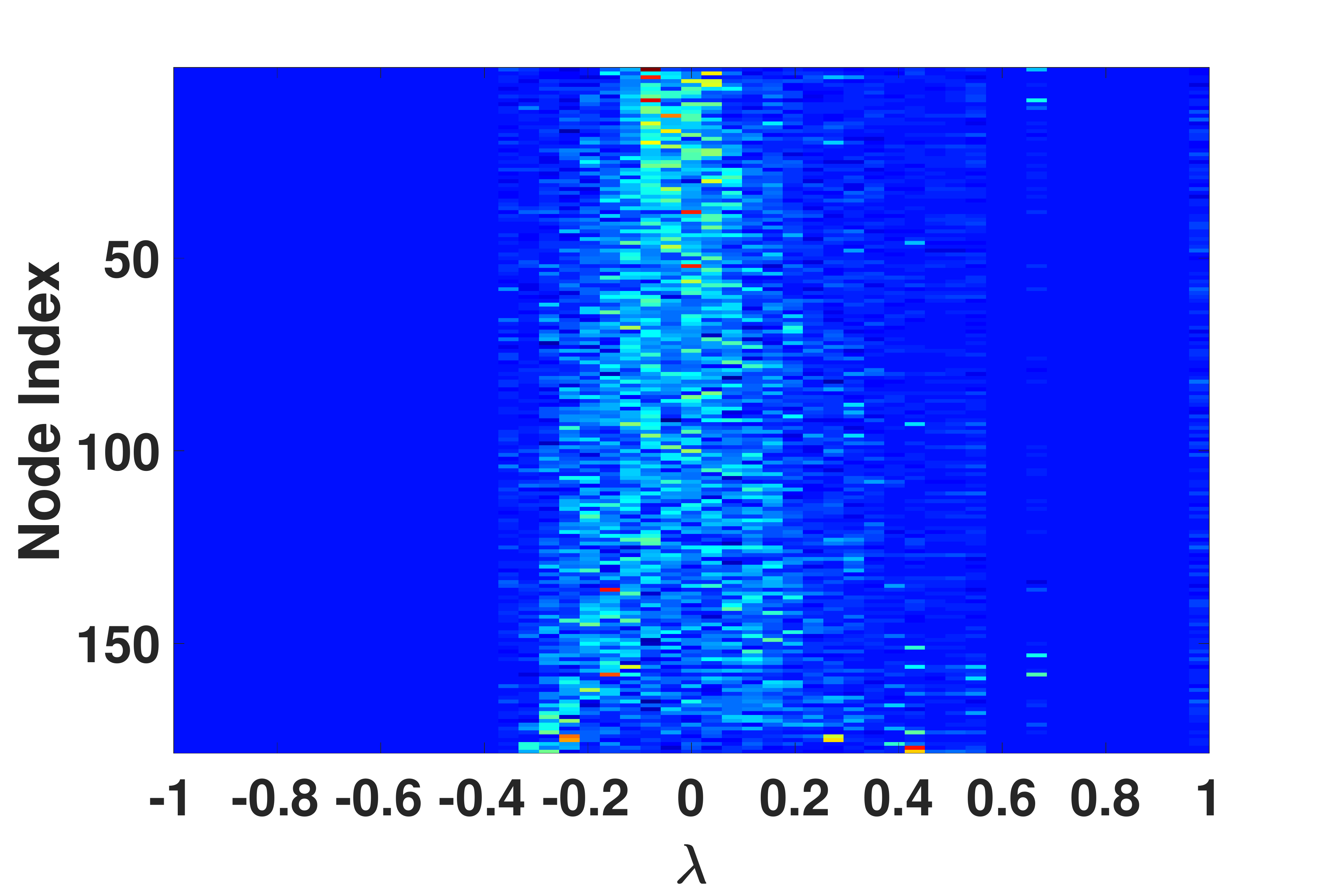}{}
    \caption{Harry Potter Characters Network}
    \label{fig:hp_ldos}
  \end{subfigure}
  \begin{subfigure}{0.19\textwidth}
    \centering  
    \captionsetup{justification=centering}
    \includegraphics[width=\textwidth,trim = .4cm 0.5cm 3.5cm 1.3cm,clip]
    {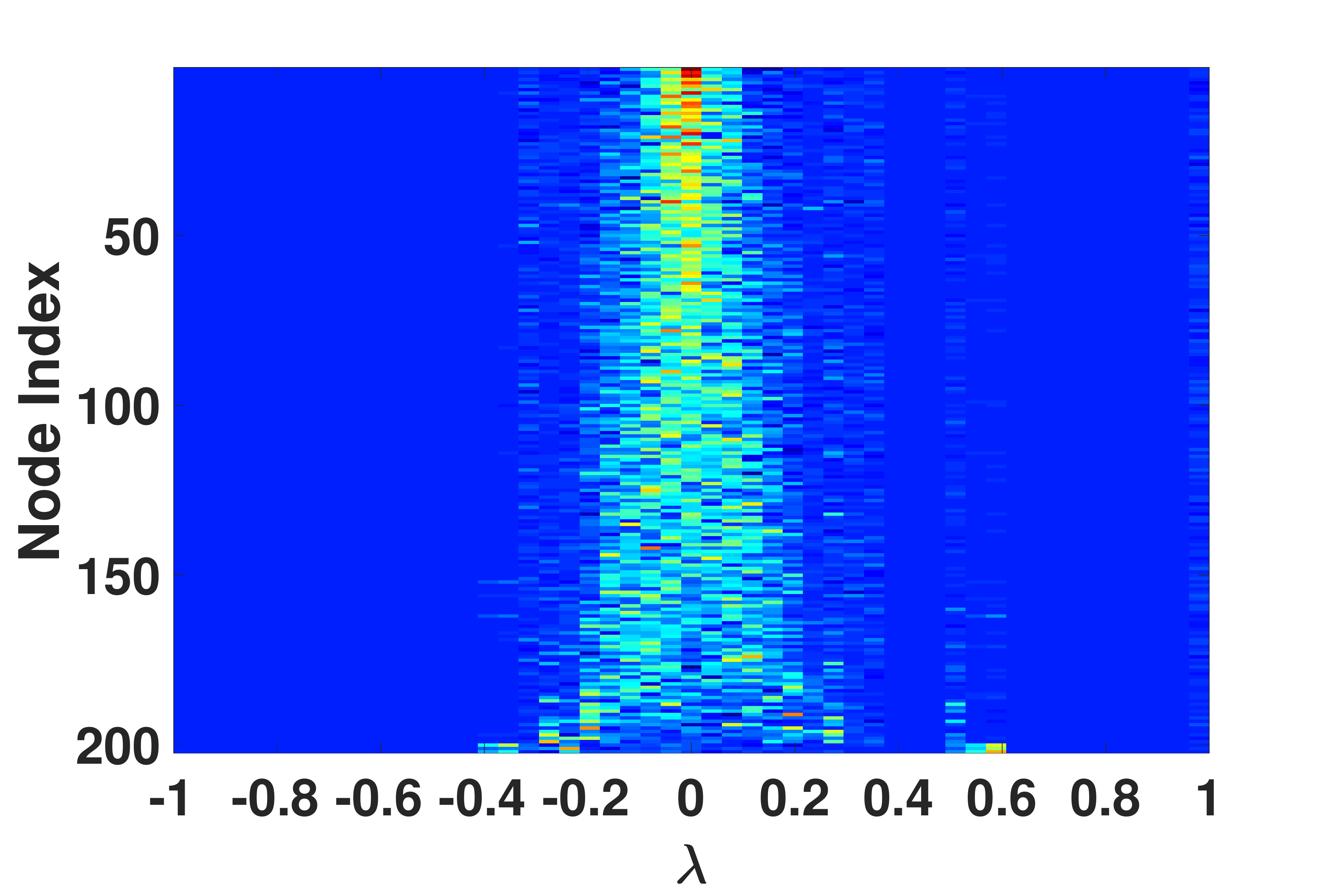}
    \caption{Twitter Ego \\Networks}
    \label{fig:twitter_ldos}
  \end{subfigure}
  \begin{subfigure}{0.19\textwidth}
    \centering  
    \captionsetup{justification=centering}
    \includegraphics[width=\textwidth,trim = .4cm 0.5cm 3.5cm 0.3cm,clip]
    {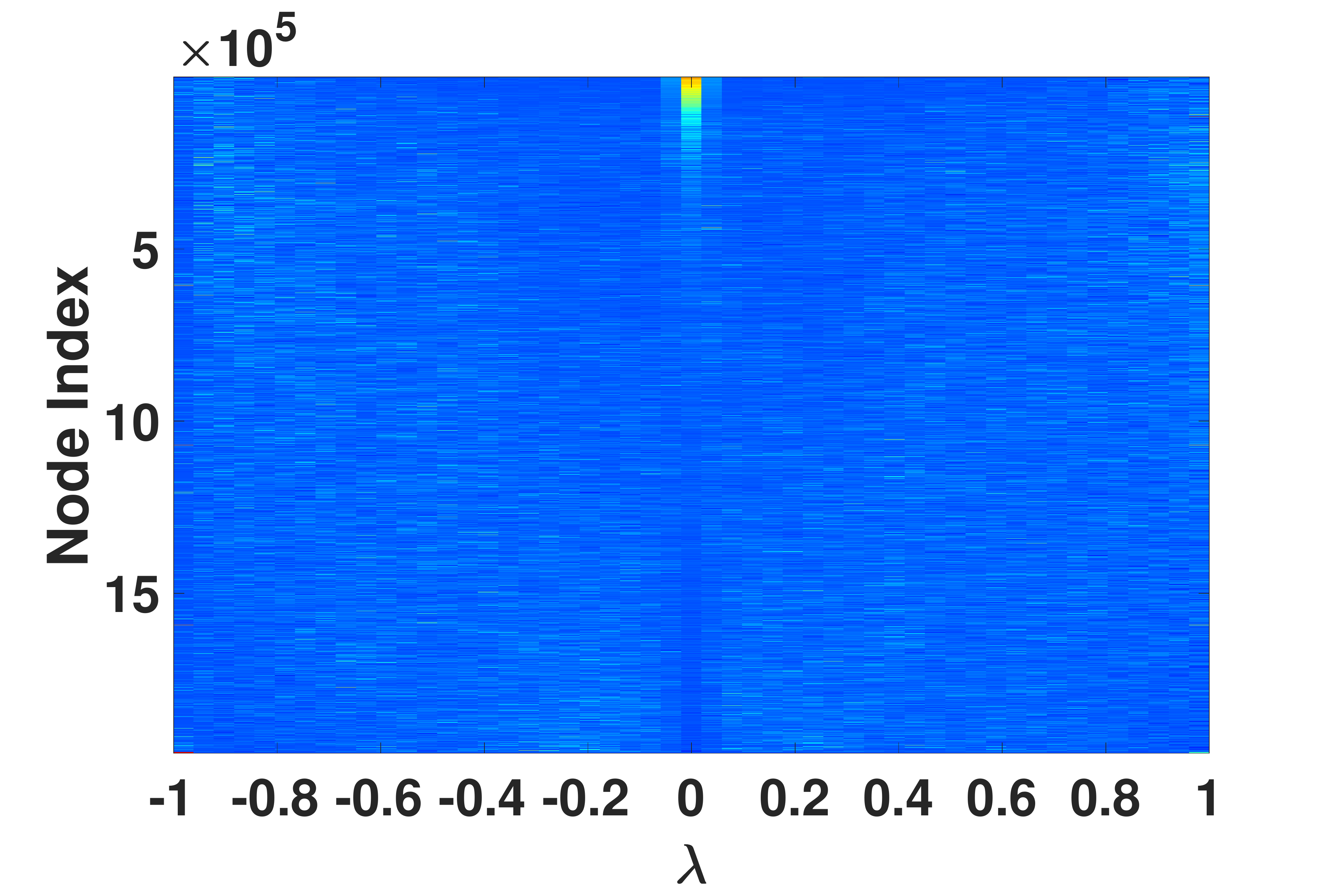}
    \caption{California Road Network}
    \label{fig:roadca_ldos}
  \end{subfigure}
  \caption{DOS(top)/PDOS(bottom) histograms for the normalized adjacency of 10
  networks from five domains. For DOS, blue bars are the true spectrum, and red
  points are from KPM ($500$ moments and $20$ Hadamard probes). For PDOS, the
  spectral histograms of all nodes are aligned vertically. Red indicates high
  weight around an eigenvalue, and blue indicates low weight. The true spectrum
  for the California Road Network (\ref{fig:roadca_ldos}) is omitted, as it is
  too large to compute exactly (1,965,206 nodes).}
  \label{fig:gallery}
\end{figure*}

\subsection{Gallery of DOS/PDOS}
We first present our spectral histogram approximation from DOS/ PDOS on a wide
variety of graphs, including collaboration networks, online social networks,
road networks and autonomous systems (dataset details are in the appendix). For
all examples, we apply our methods to the normalized adjacency matrices using
$500$ Chebyshev moments and $20$ Hadamard probe vectors. Afterwards, the
spectral density is integrated into $50$ histogram bins. In figure 
\ref{fig:gallery}, the DOS approximation is on the first row, and the PDOS
approximation is on the second. When a spike exists in the spectrum, we apply
motif filtering, and DOS is zoomed appropriately to show the remaining part. For
PDOS, we stack the spectral histograms for all nodes vertically, sorted by their
projected weights on the leading left singular vector. Red indicates that a node
has high weight at certain parts of the spectrum, whereas blue indicates low
weight.

We observe many distinct shapes of spectrum in our examples. The eigenvalues of
denser graphs, such as the Marvel characters network (\ref{fig:marvel_ldos}:
average degree $52.16$) and Facebook union of ego networks 
(\ref{fig:facebook_ldos}: average degree $43.69$), exhibit decay similar to
thepower-law around $\lambda=0$. There has been study on the power-law
distribution in the eigenvalues of the adjacency and the Laplacian matrix, but
it only focuses on the leading eigenvalues rather than the entire spectrum~
\cite{eikmeier2017revisiting} for large real-world datasets. Relatively sparse
graphs (\ref{fig:erdos_ldos}: average degree $3.06$,; \ref{fig:as_ldos}: average
degree $4.13$) often possess spikes, especially around $\lambda=0$, which
reflect a larger set of loosely-connected boundary nodes. It is much more
evident in the PDOS spectral histograms, which allow us to pick out the nodes
with dominant weights at $\lambda=0$ and those that contribute most to local
structures. Finally, though the road network is quite sparse (ave. deg $2.50$),
its regularity results in a lack of special features, and most nodes contribute
evenly to the spectrum according to PDOS.

\subsection{Computation time}

In this experiment, we show the scaling of our methods by applying them to
graphs of varying size of nodes, edges, and sparsity patterns. Rather than
computation power, the memory cost of loading a graph with 100M-1B edges is more
often the constraint. Hence, we report runtimes for a Python version on a Google
Cloud instance with 200GB memory and an Intel Xeon E5 v3 CPU at 2.30GHz.

The datasets we use are obtained from the SNAP repository~\cite{snapnets}. For
each graph, we compute the first $10$ Chebyshev moments using KPM with $20$
probe vectors. Most importantly, the cost for each moment is independent of the
total number of moments we compute. Table \ref{tab:scaling} reports number of
nodes, number of edges, average degree of nodes, and the average runtime for
computing each moment. We can observe that the runtime is in accordance with
the theoretical complexity $O(N_z(|V| + |E|))$. For the Friendster social
network with about 1.8 billion edges, computing each moment takes about 1000
seconds to compute, which means we could obtain a rough approximation to its
spectrum within a day. As the dominant cost is matrix-matrix multiplication and
we use several probe vectors, our approach has ample opportunity for parallel
computation.
\begin{table}
\centering
\caption{%
Average time to compute each Chebyshev moment (with 20 probes) for graphs from
the SNAP repository.}
\label{tab:scaling}
\begin{tabular}{r l l l l}
\toprule
Network & \# Nodes &\# Edges & Avg.\ Deg.\ & Time (s) \\ \midrule
Facebook & 4,039 & 88,234 & 43.69 & 0.007\\
AstroPh & 18,772 & 198,110 & 21.11 & 0.028\\
Enron & 36,692 & 183,831 & 10.02 & 0.046\\
Gplus & 107,614 & 13,673,453 & 254.12 & 1.133\\
Amazon & 334,863 & 925,872 & 5.53 & 0.628\\
Neuron & 1,018,524 & 24,735,503 & 48.57 & 9.138\\
RoadNetCA & 1,965,206 & 2,766,607 & 2.82 & 2.276\\
Orkut & 3,072,441 & 117,185,083 & 76.28 & 153.7\\
LiveJournal & 3,997,962 & 34,681,189 & 17.35 & 14.52 \\
Friendster & 65,608,366 & 1,806,067,135 & 55.06 & 1,017\\
\bottomrule
\end{tabular}
\end{table}

\subsection{Model Verification}

In this experiment, we investigate the spectrum for some of the popular graph
models, and whether they resemble the behavior of real-world data. Two of the
most popular models used to describe real-world graphs are the scale-free model 
\cite{barabasi1999emergence} and the small-world model 
\cite{watts1998collective}. Farkas et al. \cite{farkas2001spectra} has analyzed
the spectrum of the adjacency matrix; we instead consider the normalized
adjacency.

The scale-free model grows a random graph with the preferential attachment
process, starting from an initial seed graph and adding one node and $m$ edges
at every step. Figure \ref{fig:ba} shows spectral histograms for this model
with $5000$ nodes and different choices of $m$. When $m=1$, the generated graph
has abundant local motifs like many sparse real-world graphs. By searching in
PDOS for the nodes that have high weight at the two spikes, we find node-doubles
($\lambda=0)$ and singly-attached chains ($\lambda=\pm 1/\sqrt{2}$). When $m=5$,
the graph is denser, without any particular motifs, resulting in an
approximately semicircular spectral distribution.

\begin{figure}
  \begin{subfigure}{0.235\textwidth}
    \centering  
    \captionsetup{justification=centering}
    \includegraphics[width=\textwidth,trim = .4cm 0.5cm 3.5cm 1.3cm,clip]{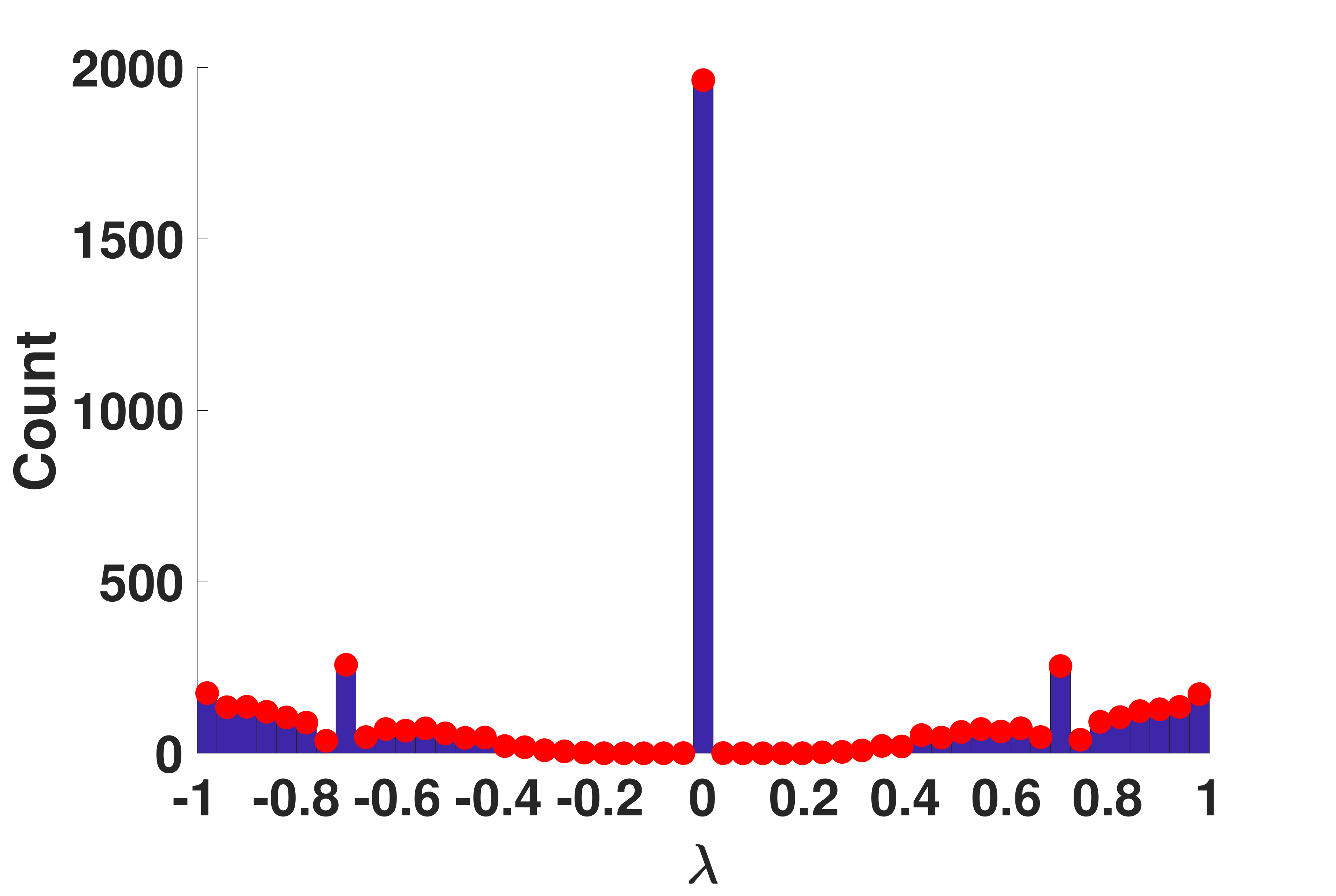}
    \caption{$m = 1$}
    \label{fig:ba_sparse}
  \end{subfigure}
  \begin{subfigure}{0.235\textwidth}
    \centering
    \captionsetup{justification=centering}
    \includegraphics[width=\textwidth,trim = .4cm 0.5cm 3.5cm 1.3cm,clip]{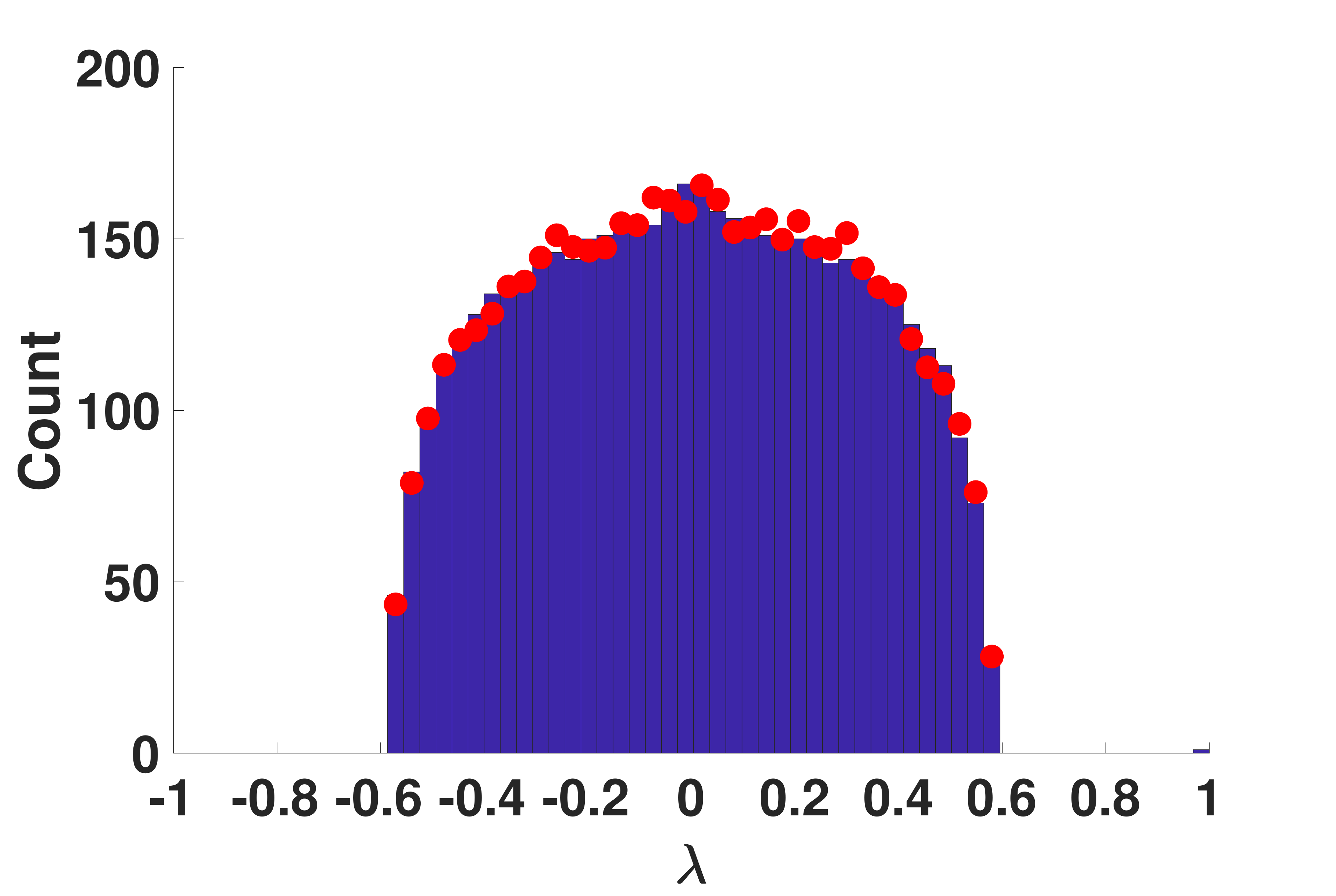}
    \caption{$m = 5$}
    \label{fig:ba_dense}
  \end{subfigure}
  \caption{Spectral histogram for scale-free model with $5000$
  nodes and different $m$. Blue 
  bars are the real spectrum, red points are from KPM
  ($500$ moments and $20$ probes).}
  \label{fig:ba}
\end{figure}

The small-world model generates a random graph by re-wiring edges of a ring
lattice with a certain probability $p$. Here we construct these
graphs on $5000$ nodes with $p=0.5$; the pattern in spectrum is
insensitive for a wide range of $p$. In Figure \ref{fig:sw}, when the graph is
sparse with $5000$ edges, the spectrum has spikes at $0$ and $\pm 1$, indicating
local symmetries, bipartite structure, and disconnected components. With $50000$
edges, localized structures disappear and the spectrum has narrower support.

\begin{figure}
  \begin{subfigure}{0.235\textwidth}
    \centering  
    \captionsetup{justification=centering}
    \includegraphics[width=\textwidth,trim = .4cm 0.5cm 3.5cm 1.3cm,clip]{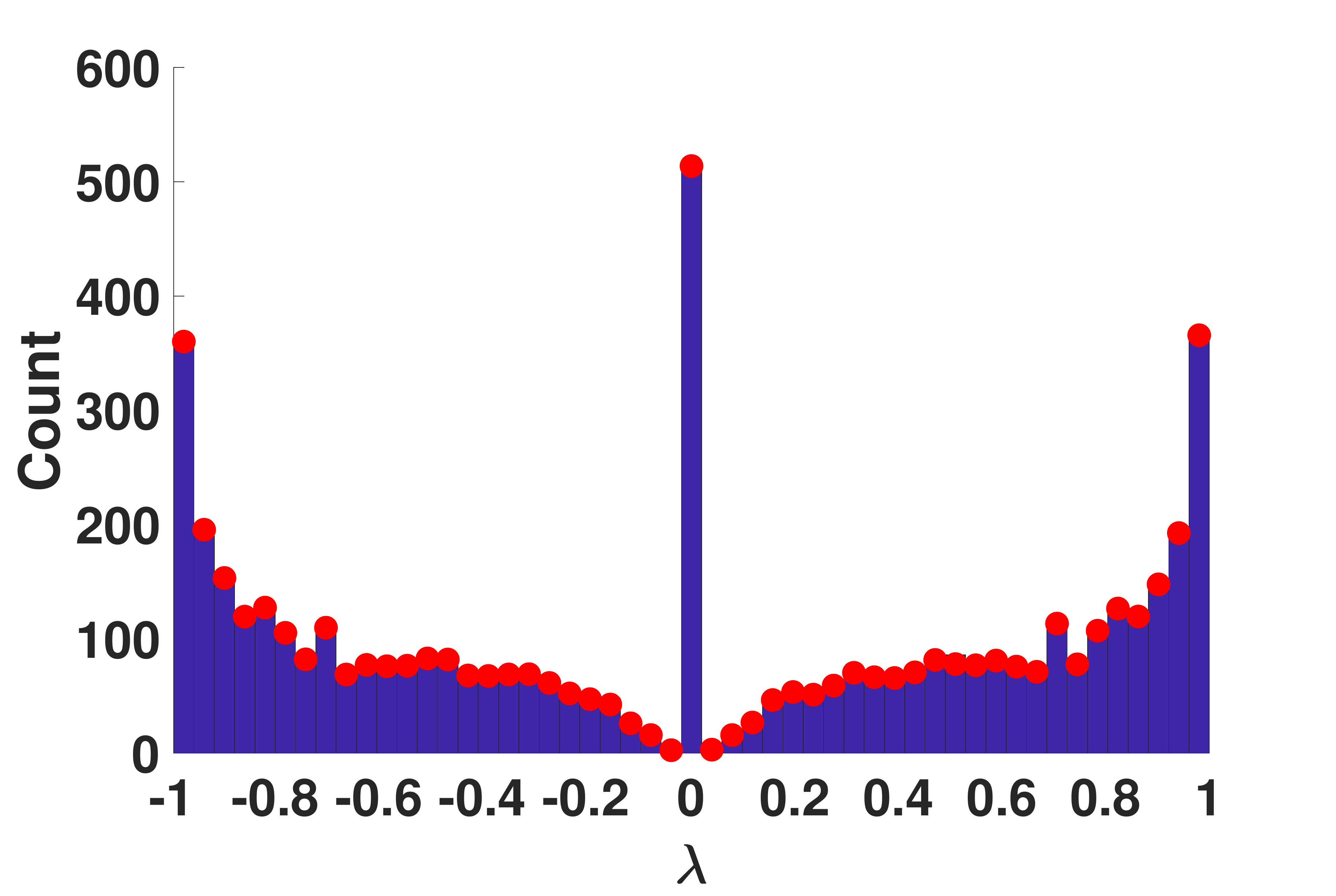}
    \caption{$|E|=5k$}
    \label{fig:sw_sparse}
  \end{subfigure}
  \begin{subfigure}{0.235\textwidth}
    \centering
    \captionsetup{justification=centering}
    \includegraphics[width=\textwidth,trim = .4cm 0.5cm 3.5cm 1.3cm,clip]{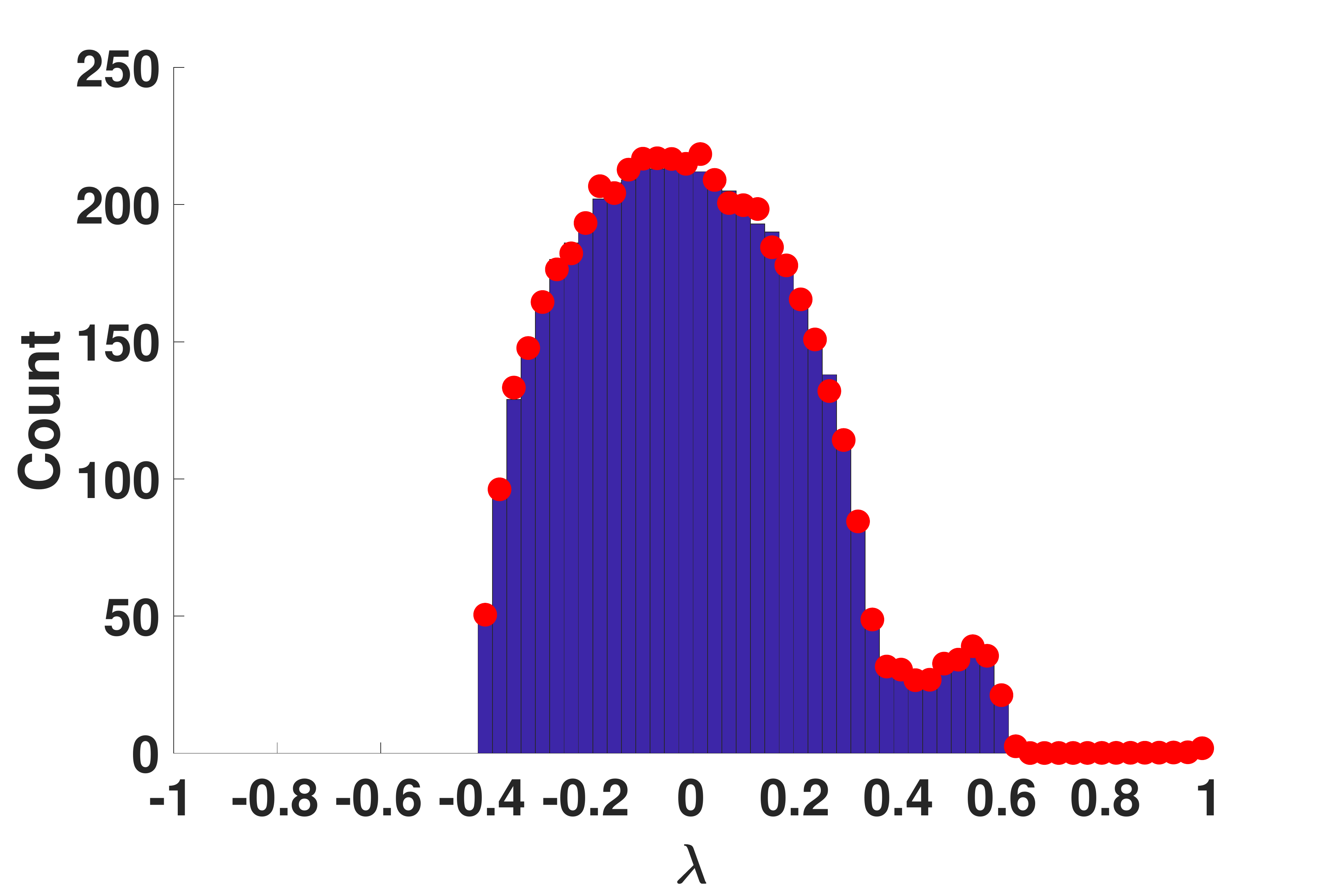}
    \caption{$|E|=50k$}
    \label{fig:sw_dense}
  \end{subfigure}
  \caption{Spectral histograms for small-world model with $5000$
  nodes and re-wiring probability $p=0.5$, starting with $5000$
  (\ref{fig:sw_sparse}) and $50000$ (\ref{fig:sw_dense} edges.
  Blue bars are the real spectrum,
  red points are from KPM ($5000$ moments and $20$ probes).}
  \label{fig:sw}
\end{figure}

Finally, we investigate the Block Two-Level Erd\"{o}s-R\'{e}nyi (BTER) model 
\cite{seshadhri2012community}, which directly fits an input graph. BTER
constructs a similar graph by a two-step process: first create a collection of
Erd\"{o}s-R\'{e}nyi subgraphs, then interconnect those using a Chung-Lu model  
\cite{chung2002connected}. Seshadhri et al. showed their model accurately
captures the observable properties of the given graph, including the eigenvalues
of the adjacency matrix. Figure \ref{fig:bter} compares the DOS/PDOS of the
Erd\"{o}s collaboration network and its BTER counterpart. Unlike the  original
graph, most $0$ eigenvalues in BTER graph come from isolated nodes. The BTER
graph also has many more isolated edges ($\lambda=\pm1$), singly-attached chains
($\lambda=\pm1/\sqrt{2})$), and singly-attached triangles ($\lambda=-1/2$). We
locate these motifs by inspecting nodes with high weights at respective part of
the spectrum.

\begin{figure}
  \begin{subfigure}{0.235\textwidth}
    \centering  
    \captionsetup{justification=centering}
    \includegraphics[width=\textwidth,trim = .4cm 0.5cm 3.5cm 1.3cm,clip]
    {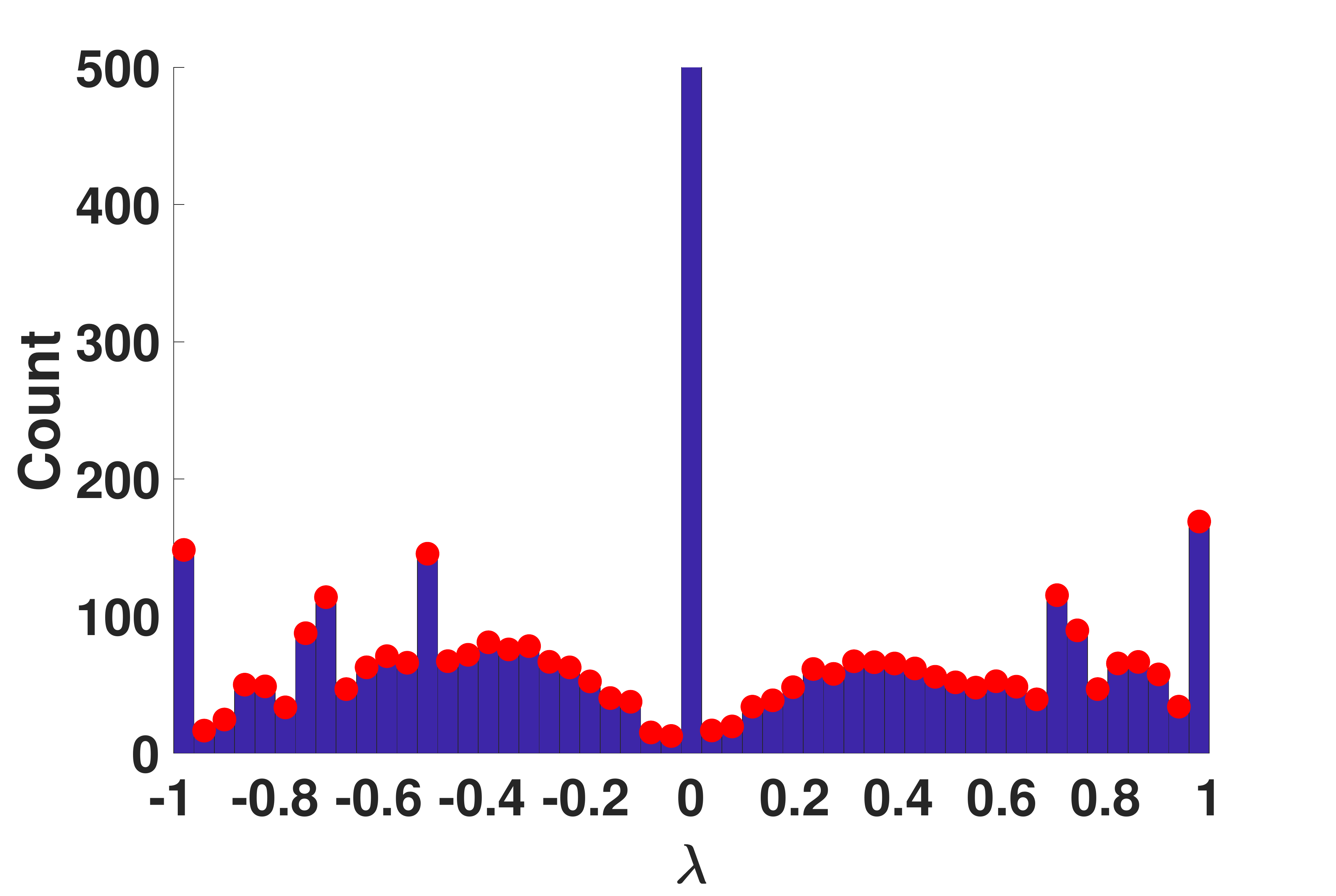}
    \caption{BTER DOS}
    \label{fig:bter_dos}
  \end{subfigure}
  \begin{subfigure}{0.235\textwidth}
    \centering
    \captionsetup{justification=centering}
    \includegraphics[width=\textwidth,trim = .4cm 0.5cm 3.5cm 1.3cm,clip]
    {./erdos}
    \caption{Erd\"{o}s DOS}
    \label{fig:erdos_dos2}
  \end{subfigure}
  \begin{subfigure}{0.235\textwidth}
    \centering  
    \captionsetup{justification=centering}
    \includegraphics[width=\textwidth,trim = .4cm 0.5cm 3.5cm 1.3cm,clip]
    {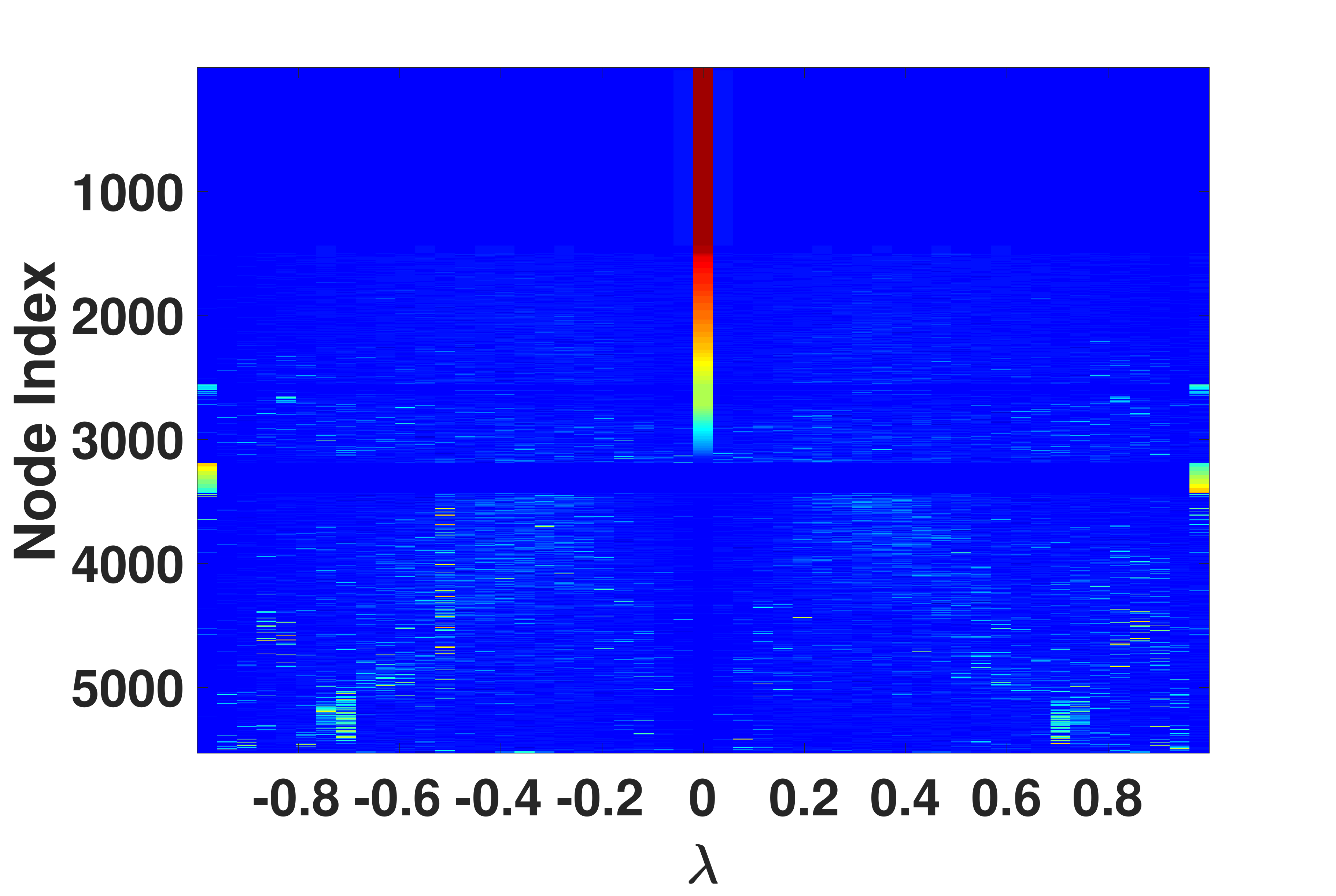}
    \caption{BTER PDOS}
    \label{fig:bter_ldos}
  \end{subfigure}
  \begin{subfigure}{0.235\textwidth}
    \centering
    \captionsetup{justification=centering}
    \includegraphics[width=\textwidth,trim = .4cm 0.5cm 3.5cm 1.3cm,clip]
    {./erdos_ldos}
    \caption{Erd\H{o}s PDOS}
    \label{fig:erdos_ldos2}
  \end{subfigure}
  \caption{Comparison of spectral histogram between Erd\H{o}s Collaboration
  Network and the BTER model. Both DOS and PDOS are computed with $500$ moments
  and $20$ probe vectors.} \label{fig:bter}
\end{figure}

\section{Discussion}
\label{sec:discussion}

In this paper, we make the computation of spectral densities a
practical tool for the analysis of large real-world network.  Our
approach borrows from methods in solid state physics, but with
adaptations that improve performance in the network analysis setting
by special handling of graph motifs that leave distinctive spectral
fingerprints.  We show that the spectral densities are stable to small
changes in the graph, as well as providing an analysis of the approximation
error in our methods.  We illustrate the efficiency of our approach by
treating graphs with tens of millions of nodes and billions of edges
using only a single compute node.  The method provides a compelling
visual fingerprint of a graph, and we show how this fingerprint can be
used for tasks such as model verification.

Our approach opens the door for the use of complete spectral
information in large-scale network analysis.  It provides a
framework for scalable computation of quantities 
already used in network science, such as common centrality
measures and graph connectivity indices (such as the Estrada index)
that can be expressed in terms of the diagonals and traces of matrix
functions.  But we expect it to serve more generally to define new
families of features that describe graphs and the roles nodes play
within those graphs. We have shown that graphs from different 
backgrounds demonstrate distinct spectral characteristics, and thus 
can be clustered based on those. Looking at LDOS across nodes for role
discovery, we can identify the ones with high similarity in their local 
structures. Moreover, extracting nodes with large weights at various points of
the spectrum uncovers motifs and symmetries. In the future, we expect to use
DOS/LDOS as graph features for applications in graph clustering, graph
matching, role classification, and other tasks.

\textbf{Acknowledgments.} We thank NSF DMS-1620038 for supporting this work. 

\bibliographystyle{ACM-Reference-Format}
\bibliography{references,arb-refs}

\newpage
\appendix
\section{Data Source}

The datasets used in this paper mainly come from the SNAP \cite{snapnets} and
RODGER repositories~\cite{rodger}. Table \ref{tab:data} is a list of the
networks from these two sources.

\begin{table}[H]
\centering
\caption{List of datasets and the corresponding source.}
\label{tab:data}
\begin{tabular}{r l l}
\toprule
Network & Full Name & Source \\ \midrule
Facebook & Facebook Ego Networks & SNAP\\
Gplus & Google+ Ego Networks & SNAP\\
Twitter & Twitter Ego Networks & SNAP\\
LiveJournal & LiveJournal Online Social Network & SNAP \\
Friendster & Friendster Online Social Network & SNAP\\
Orkut & Orkut Online Social Network & SNAP\\
Amazon & Amazon Product Network & SNAP\\
Enron & Enron Email Communication Network & SNAP\\
AstroPh & Arxiv Astro Physics Collaboration \\
        & Network & SNAP\\
HepTh & Arxiv High Energy Physics Theory \\
      & Collaboration Network & SNAP\\
RoadNetCA & California Road Network & SNAP\\
AS-733 & Autonomous System Network & SNAP\\
AS-CAIDA& CAIDA Autonomous System Network & SNAP\\
Neuron & Megascale Cell-Cell Similarity Network & SNAP\\
Erd\"{o}s & Erd\"{o}s Collaboration Network & RODGER\\
Marvel Chars & Marvel Characters Network & RODGER\\
\bottomrule
\end{tabular}
\end{table}

In addition, we used the Minnesota Road Network from the SuiteSparse Matrix
Collection \cite{davis2014suitesparse}, and the Harry Potter Characters Network
from an open source repository \cite{hprepo}.

\section{Code Release}

Code for reproducible experiments and figures are available at 
\url{https://github.com/kd383/NetworkDOS}.

\end{document}